\documentclass[preprint,12pt,english]{article}
\usepackage{amssymb,amsmath,amsthm,tikz,xcolor,float}
\def\tpl{\mathtt{tpl}}
\def\sz{\mathtt{sz}}
\def\xpm{\mathtt{dh}}
\def\hole{\bullet}
\def\tr{\mathtt{r}}
\def\tty{\mathtt{st}}

\def\tS{\mathtt{S}}

\def\tA{\mathtt{A}}

\def\cE{\mathcal{E}}

\def\cA{\mathcal{A}}
\def\cC{\mathcal{C}}
\def\cF{\mathcal{F}}
\def\cL{\mathcal{L}}
\def\cT{\mathcal{T}}
\def\cQ{\mathcal{Q}}
\def\sk{\mathtt{sk}}
\def\depth{\mathtt{d}}

\def\Ra{\Rightarrow}
\newcommand\pair[1]{\langle{#1}\rangle}
\newcommand\comment[1]{}

\def\ty{\mathtt{t}}
\def\todo{\begin{center}{\color{red}\tt [TODO]}\end{center}}

\newtheorem{definition}{Definition}
\newtheorem{example}{Example}
\newtheorem{lemma}{Lemma}
\newtheorem{corollary}{Corollary}
\newtheorem{theorem}{Theorem}

\begin{document}

\title{Learning cover context-free grammars from structural data}

\author{ Mircea Marin, 
Gabriel Istrate \footnote {Dept. of Computer Science, West University of Timi\c{s}oara and e-Austria Research Institute, Bd. V. P\^{a}rvan 4, cam. 045
B, Timi\c{s}oara, RO-300223, Romania. Corresponding author's email: {\tt mmarin@info.uvt.ro }}}

\maketitle
\begin{abstract}
We consider the problem of learning an unknown context-free grammar when the only knowledge available and of interest to the learner is about its structural descriptions with  depth at most $\ell.$ The goal is to learn a {\em cover context-free grammar} (CCFG) with respect to $\ell$, that is, a CFG whose structural descriptions with depth at most $\ell$ agree with those of the unknown CFG. 
We propose an algorithm, called $LA^\ell$,  that efficiently learns a CCFG using two types of queries: structural equivalence  and structural membership. 
We show that  $LA^\ell$ runs in time polynomial in the number of states of a minimal deterministic finite cover tree automaton (DCTA) with respect to $\ell$. This number is often much smaller than the number of states of a minimum deterministic finite tree automaton for the structural descriptions of the unknown grammar. 
\end{abstract}


\section{Introduction}

Angluin's approach to grammatical inference \cite{Angluin:87} is an important contribution to computational learning, with extensions to problems, such as compositional verification and synthesis \cite{farzan2008extending,maler1995learnability}, that go beyond the usual applications to natural language processing and computational biology \cite{de2010grammatical}.

Practical concerns, e.g. \cite{kumar2006minimization}, seem to require going beyond regular languages to classes of languages with regular tree nature.  However, Angluin and Kharitonov have shown that learning CFGs from membership and equivalence queries is intractable under plausible cryptographic assumptions \cite{angluin1995won}. A way out is to learn structural descriptions of context free languages. Sakabibara has shown that Angluin's algorithm extends to this setting \cite{Sak:1990}. His approach has applications in learning the structural descriptions of natural languages, which describe the shape of the parse trees of well chosen CFGs. Often, these structural descriptions are subject to additional restrictions arising from modelling considerations. For instance, in natural language understanding, the bounded memory restriction on human comprehension seems to limit the recursion depth of such a parse tree to a constant. A natural example with a similar flavour is the limitation imposed by the \LaTeX\ system, that limits the number of nestings of itemised environments to a small constant.

Imposing such a restriction leads to the idea of learning cover languages, that is, languages that accurate up to an equivalence. For regular languages modulo a finite prefix such an approach has been pursued by Ipate \cite{Ipate:2012} (see also \cite{holzer2012equivalence}).

\comment{
We address the problem of learning an unknown context-tree grammar (CFG) when the only information available and of interest to the learner are its structural descriptions with depth at most $\ell$. The structural descriptions of a CFG are the trees obtained from the derivation trees of the grammar by unlabelling all its internal nodes. They have been recognised by Levy and Joshi as useful data for the efficient learning of the unknown grammar \cite{LJ:1978}. 

Using the fact that, for a CFG, the set of derivation trees  and the set of structural descriptions  are rational sets of trees, Sakabibara proposed in \cite{Sak:1990} an algorithm, called $LA$, which reduces the problem of learning a CFG from structural descriptions to that of learning a tree automaton. His algorithm is an extension of Angluin's efficient learning algorithm $L^*$ for finite automata  \cite{Angluin:87} to  one for tree automata, using two types of queries: structural equivalence queries and structural membership queries. $LA$  runs in time polynomial in the number of states of a minimal deterministic bottom-up tree automaton for the structural descriptions of the unknown grammar and the maximum size of any counterexample returned by a structural equivalence~query.

Algorithms $L^*$ and $LA$ are suitable for learning efficiently an {\em exact} characterisation of the (regular or context-free) language under consideration. But, as noticed by Ipate in \cite{Ipate:2012}, there are many practical applications where characterisations of finite subsets of the language suffice. For regular languages where the interest is only in words with length at most $\ell$, Ipate defines the notion of minimal {\em deterministic finite cover automaton} (DFCA) with respect to $\ell$, and proposes an algorithm to learn such an DFCA in time polynomial in the number or its states~\cite{Ipate:2012}. His algorithm, called $L^\ell$,  is a nontrivial adaptation of algorithm~$L^*$.

Ipate's ideas can be carried over to context-free languages if we restrict the learning to a CFG whose structural descriptions are of interest only when they are of depth at most $\ell$. In this case, we are interested to learn a {\em cover CFG}, that is, a context-free grammar whose structural descriptions with depth at most $\ell$ coincide with the structural descriptions of the unknown CFG; structural descriptions of higher depth can also be generated, but they are of no interest to the learning process. Like Sakakibara, we reduce the problem of learning such a CFG to that of learning a special tree automaton which, in our case, turns out to be a minimal deterministic finite cover tree automaton (DCTA) for the structural descriptions of interest. Formally, a minimal DCTA of an unknown context-free grammar $G_U$ with respect to $\ell$ is a minimal deterministic bottom-up tree automaton which accepts all structural descriptions of $G_U$ with depth at most $\ell$, and possibly other with depth larger than $\ell$. This kind of tree automaton is an adaptation of Ipate's minimal DFCA to languages of structural descriptions.

a\\
a\\
a\\
a\\
a
}

In this paper, we extend this approach to context-free languages with structural descriptions. 
We propose an algorithm called $LA^\ell$ which 
\comment{is an adaptation of the learning algorithms $LA$ and $L^\ell$ in the following ways:
It} 
asks two types of queries: structural equivalence and structural membership queries, both restricted to structural descriptions with depth at most $\ell$, where $\ell$ is a constant.
$LA^\ell$ stores the answers retrieved from the teacher in an {\em observation table} which is used to guide the learning protocol and to construct a minimal DCTA of the unknown context-free grammar with respect to $\ell$.
Our main result shows that $LA^\ell$ runs in time  polynomial in $n$ and $m$, where $n$ is the number of states of a minimal DCTA of the unknown CFG with respect to $\ell$, and $m$ is the maximum size of a counterexample returned by a failed structural membership query. 

The paper is structured as follows. Section \ref{prelim} introduces the basic notions and results to be used later in the paper. It also describes algorithm $LA$. 
In Sect. \ref{sect4} we introduce the main concepts related to the specification and analysis of our learning algorithm $LA^\ell$. They are natural generalisations to languages of structural descriptions of the concepts proposed by Ipate \cite{Ipate:2012} in the design and study of his algorithm $L^\ell$. 
In Sect. \ref{sect5} we analyse the space and time complexity of $LA^\ell$ and show that its time complexity is a polynomial in $n$ and $m$, where $n$ is the number of states of a minimal deterministic finite cover automaton w.r.t. $\ell$ of the language of structural descriptions of interest, and $m$ is an upper bound to the size of counterexamples returned by failed structural equivalence queries. 
\section{Preliminaries}
\label{prelim}
We write $\mathbb{N}$ for the set of nonnegative integers, $A^*$ for the set of finite strings over a set $A$, and $\epsilon$ for the empty string. If $v,w\in A^*$, we write $v\leq w'$ if there exists $w'\in A^*$ such that $v w'=w$; $v<v'$ if $v\leq v'$ and $v\neq v'$; and $v\perp w$ if neither $v\leq w$ nor $w\leq v.$  
\subsection*{Trees, terms, contexts, and context-free grammars}
A {\em ranked alphabet} is a finite set $\cF$ of function symbols  together with a finite {\em rank} relation $rk(\cF)\subseteq \cF\times \mathbb{N}.$ We denote the subset $\{f\in\cF\mid (f,m)\in rk(\cF)\}$ by $\cF_m$,  the set $\{m\mid (f,m)\in rk(\cF)\}$ by $ar(f)$, and  $\bigcup_{f\in \cF}ar(f)$ by $ar(\cF)$.
The {\em terms} of the set $\cT(\cF)$ are the strings of symbols defined recursively by the grammar
$t\mathop{\mbox{\tt{::=}}}a\mid f(t_1,\ldots,t_m)$
where $a\in\cF_0$ and $f\in\cF_m$ with $m>0$. The {\em yield} of a term $t\in\cT(\cF)$ is the finite string $yield(t)\in\cF_0^*$ defined as follows: $yield(a):=a$ if $a\in\cF_0$, and $yield(f(t_1,\ldots,t_m)):=w_1\ldots w_m$ where $w_i=yield(t_i)$ for $1\leq i\leq m.$

A {\em finite ordered tree} over a set of labels $\cF$ is a mapping $t$ from a nonempty and prefix closed set $Pos(t)\subseteq(\mathbb{N}\setminus\{0\})^*$ into $\cF$. Each element in $Pos(t)$ is called a {\em position}. The tree $t$ is {\em ranked} if $\cF$ is a ranked alphabet, and $t$ satisfies the following additional property:
For all $p\in Pos(t)$, there exists $m\in\mathbb{N}$ such that  $\{i\in\mathbb{N}\mid pi\in Pos(t)\}=\{1,\ldots,m\}$ and  $t(p)\in \cF_m$.

Thus, any term $t\in \cT(\cF)$ may be viewed as a finite ordered ranked tree, and we will  refer to it by ``tree'' when we mean the finite ordered tree with the additional property mentioned above. The {\em depth} of $t$ is $\depth(t):=\max \{\|p\|\mid p\in Pos(t)\}$ where $\|p\|$ denotes the length of $p$ as sequence of numbers. The {\em size} $\sz(t)$ of  $t$ is the number of elements of the set $\{p\in Pos(t)\mid \| p\| \neq\depth(t)\}$, that is, the number of internal nodes of $t$.

The {\em subterm} $t|_p$ of a term $t$ at position $p\in Pos(t)$ is defined by the following: $Pos(t|_p):=\{i\mid pi\in Pos(t)\}$, and $t|_p(q):=t(pp')$ for all $p'\in Pos(t|_p).$ 
We denote by $t[u]_p$ the term obtained by replacing in $t$ the subterm $t|_p$ with $u$, that is: $Pos(t[u]_p)=(Pos(t)-\{pp'\mid p'\in Pos(t|_p)\})\cup\{pp''\mid p''\in Pos(u)\}$, and $$t[u]_p(p'):=\left\{\begin{array}{ll}
u(p'')&\mbox{if }p'=pp''\text{ with }p''\in Pos(u),\\
t(p')&\mbox{otherwise.}
\end{array}\right.$$
The set $\cC(\cF)$ of {\em contexts} over $\cF$ is the set of terms  over $\cF\cup\{\bullet\}$, where:
\begin{itemize}
\item $\bullet$ is a distinguished fresh symbol with $ar(\hole)=\{0\}$, called {\em hole},
\item $rk(\cF\cup\{\bullet\})=rk(\cF)\cup\{(\bullet,0)\}$, and 
\item every element $C\in\cC(\cF)$ contains only one occurrence of $\bullet$. This is the same as saying that $\{p\in Pos(C)\mid C=C[\bullet]_p\}$ is a singleton set. 
\end{itemize}
If $C\in\cC(\cF)$ and $u\in\cC(\cF)\cup\cT(\cF)$ then $C[u]$ stands for the context or term $C[u]_p$, where $C=C[\bullet]_p.$ The {\em hole depth} of a context $C\in\cC(\cF)$ is $\depth_\hole(C):=\|p\|$ where $p$ is the unique position of $C$ such that $C=C[\hole]_p.$ 
From now on, whenever $M$ is a set of terms, $P$ is a set of contexts, and $m$ is a non-negative integer, we define the sets
$M_{[m]}:=\{t\in M\mid \depth(t)\leq m\}$ and $P_{\pair{m}}:=\{C\in P\mid \depth_\hole(C)\leq m\}$.

\comment{
A {\em context-free grammar} (CFG) is a quadruple $G=(N,\Sigma,P,S)$ where: $N$ is a finite set of {\em nonterminals}; $\Sigma$ is a finite set of {\em terminals}; $P$ is a finite set of {\em productions}  $X\to\alpha$ with $X\in N$ and $\alpha\in (N\cup \Sigma)^*$; and $S\in N$ is a special nonterminal called the {\em start symbol}.  Such a CFG induces a {\em derivation relation} $\Ra_G$ on the set of finite strings $(N\cup \Sigma)^*$, defined by $\alpha_1X\alpha_2\Ra_G\alpha_1\alpha\alpha_2$ whenever  $\alpha_1,\alpha_2\in (N\cup \Sigma)^*$ and $(X\to\alpha)\in P.$ The {\em language generated} by $G$ is $\cL(G):=\{\alpha\in \Sigma^*\mid S\Ra_G^* \alpha\}$, where $\Ra^*_G$ is the reflexive and transitive closure of $\Ra_G$. A {\em context-free language} is a language generated by a CFG.}
We assume that the reader is acquainted with the notions of CFG and the context-free language $\cL(G)$ generated by a CFG $G$, see, e.g., \cite{Sipser:2006}.
A CFG is $\epsilon$-free if it has no productions of the form $X\to\epsilon$. It is well known~\cite{Hopcroft:2003} that every $\epsilon$-free context-free language $L$ (that is, $\epsilon\not\in L$) is generated by an $\epsilon$-free CFG. 
The derivation trees of an $\epsilon$-free CFG $G=(N,\Sigma,P,S)$  correspond to terms from $\cT(N\cup \Sigma)$ with $ar(a)=\{0\}$ for al $a\in\Sigma$ and $ar(X)=\{m\mid \exists (X\to \alpha)\in P$ with $\|\alpha\|=m\}$ for all $X\in N$. The sets $D_G(U)$ of derivation trees  issued from $U\in N\cup \Sigma$, and $D(G)$ of derivation trees of $G$, are defined recursively as follows:
\begin{align*}
 D_G(a)&:=\{a\}\text{ if }a\in \Sigma,\\
D_G(X)&:=\bigcup_{(X\to U_1\ldots U_m)\in P}\{X(t_1,\ldots,t_m)\mid t_1\in D_G(U_1)\wedge\ldots\wedge t_m\in D_G(U_m)\},
\end{align*}
$D(G):=D_G(S).$ Note that $\cL(G)=\{yield(t)\mid t\in D(G)\}.$
\subsection*{Structural descriptions and cover context-free grammars}
A {\em skeletal alphabet} is a ranked alphabet  $Sk=\{\sigma\}$, where $\sigma$ is a special symbol with $ar(\sigma)$ a finite subset of $\mathbb{N}\setminus\{0\}$, and a {\em skeletal set} is a ranked alphabet $Sk\cup A$ where $Sk\cap A=\emptyset$ and $ar(a)=0$ for all $a\in A$. 
Skeletal alphabets are intended to describe the structures of the derivation trees of $\epsilon$-free CFGs. For an $\epsilon$-free CFG $G=(N,\Sigma,P,S)$ we consider the skeletal alphabet $Sk$ with $ar(\sigma):=\{\|\alpha\| \mid (X\to\alpha)\in P\}$, and the skeletal set $Sk\cup\Sigma.$
The {\em skeletal} (or {\em structural}) {\em description} of a derivation tree $t\in D_G(U)$ is the term $\sk(t)\in\cT(Sk\cup\Sigma)$ where 
$$\sk(t):=\left\{\begin{array}{ll}
a&\mbox{if }t=a\in \Sigma,\\
\sigma(\sk(t_1),\ldots,\sk(t_m))&\text{if }t=X(t_1,\ldots,t_m)\text{ with }m>0.
\end{array}\right.$$ 
For example, if $G$ is the grammar $(\{\tS,\tA\},\{a,b\},\{\tS\to \tA,\tA\to a\tA b,\tA\to ab\},\tS)$ then $t=\tS(\tA(a,\tA(a,b),b))\in D_G(\tS)$ and $\sk(t)=\sigma(\sigma(a,\sigma(a,b),b))\in\cT(\{\sigma,a,b\})$, where $ar(\sigma)=\{1,2,3\}$ and $ar(a)=ar(b)=\{0\}$. Graphically,  we have
\begin{center}
$t=$\begin{tikzpicture}[scale=.4,baseline=-3.3em]
\node {$S$}
  child {node {$\tA$}
      child {node {$a$}}
      child {node {$\tA$}
        child {node {$a$}}
        child {node {$b$}}
      }
      child {node {$b$}}
  };
\end{tikzpicture}
\quad\qquad$\Rightarrow\quad\sk(t)=$
\begin{tikzpicture}[scale=.4,,baseline=-3.3em]
\node {$\sigma$}
  child {node {$\sigma$}
      child {node {$a$}}
      child {node {$\sigma$}
        child {node {$a$}}
        child {node {$b$}}
      }
      child {node {$b$}}
  };
\end{tikzpicture}
\end{center}
If $M$ is a set of ranked trees, the set of its structural descriptions is  $K(M):=\{\sk(t)\mid t\in M\}$. Two context-free grammars $G_1$ and $G_2$ over the same alphabet of terminals are {\em structurally equivalent} if $K(D(G))=K(D(G')).$

\begin{definition}[cover CFG]
Let $\ell$ be a positive integer and $G_U$ be an $\epsilon$-free CFG of a language $U\subseteq\Sigma^*$. 
A {\em cover context-free grammar} of $G_U$ with respect to $\ell$ is an $\epsilon$-free CFG $G'=(N,\Sigma,P,S)$ such that $K(D(G'))_{[\ell]}= K(D(G_U))_{[\ell]}$. 
\end{definition}

\subsection*{Tree automata}

The definition of tree automaton presented here is equivalent with that given in~\cite{Sak:1990}. It is non-standard in the sense that it cannot accept any tree of depth 0. 
\begin{definition}
A {\em  nondeterministic (bottom-up) finite tree automaton} (NFTA) over $\cF$ is a quadruple $\cA=(\cQ,\cF,\cQ_{\mathtt{f}},\Delta)$ where $\cQ$ is a finite set of {\em states}, $\cQ_{\mathtt{f}}\subseteq \cQ$ is the set of {\em final states}, and $\Delta$ is a set of {\em transition rules} of the form $f(q_1,\ldots,q_m)\to q$ where $m\geq 1$, $f\in\cF_m$, $q_1,\ldots,q_m\in\cF_0\cup\cQ,$ and $q\in\cQ$. 
\end{definition}
Such an automaton $\cA$ induces a {\em move} relation $\to_\cA$ on the set of terms $\cT(\cF\cup\cQ)$ where $ar(q)=\{0\}$ for all $q\in\cQ$, as follows: 
\begin{itemize}
\item[] $t\to_{\cA} t'$ if there exist $C\in \cC(\cF\cup\cQ)$ and $f(q_1,\ldots,q_m)\to q\in \Delta$ such that $t=C[f(q_1,\ldots,q_m)]$ and $t'=C[q].$
\end{itemize}
The {\em language accepted by $\cA$} is $\cL(\cA):=\{t\in\cT(\cF)\mid t\to^*_\cA q\text{ for some }q\in\cQ_{\mathtt{f}}\}$ where $\to^*_\cA$ is the reflexive-transitive closure of $\to_\cA$. In this paper, a {\em regular tree language} is a language accepted by such an NFTA.
Two NFTAs are {\em equivalent} if they accept the same language.

$\cA=(\cQ,\cF,\cQ_{\mathtt{f}},\Delta)$ is {\em deterministic} (DFTA) if the transition rules of $\Delta$ describe a mapping $\delta$ which assigns to every $m\in ar(\cF)$ a function $\delta_m$ such that
$\delta_0:\cF_0\to\cF_0$, $\delta_0(a)=a$ for all $a\in\cF_0$, and 
 $\delta_m:\cF_m\to (\cF_0\cup \cQ)^m\to \cQ$ if $m>0.$
This implies that $f(q_1,\ldots,q_m)\to q\in\Delta$ if and only if $\delta_m(f)(q_1,\ldots,q_m)=q.$ The  extension $\delta^*$ of $\{\delta_m\mid m\in ar(\cF)\}$ to $\cT(\cF)$ is defined as expected: $\delta^*(a)=a$ if $a\in\cF_0$, and $\delta^*(f(t_1,\ldots,t_m)):=\delta_m(f)(\delta^*(t_1),\ldots,\delta^*(t_m))$ otherwise. Note that, if $\cA$ is a DFTA then $\cL(\cA)=\{t\in\cT(\cF)\mid \delta^*(t)\in\cQ_{\mathtt{f}}\}.$

Two DFTAs $\cA_1=(\cQ,\cF,\cQ_{\mathtt{f}},\delta)$ and $\cA_2=(\cQ',\cF,\cQ'_{\mathtt{f}},\delta')$ are {\em isomorphic} if there exists a bijection $\varphi:\cQ\to \cQ'$ such that $\varphi(\cQ_{\mathtt{f}})=\cQ'_{\mathtt{f}}$ and for every $f\in\cF_m$, $q_1,\ldots,q_m\in\cF_0\cup\cQ$, $\varphi(\delta_m(f)(q_1,\ldots,q_m))=\delta_m'(f)(\varphi(q_1),\ldots,\varphi(q_m)).$ A {\em minimum DFTA} of a regular tree language $L\subseteq\cT(\cF)\setminus\cF_0$ is a DFTA $\cA$ with minimum number of states such that $\cL(A)=L.$

There is a strong correspondence between tree automata and $\epsilon$-free CFGs.
The NFTA corresponding to an $\epsilon$-free CFG $G=(N,\Sigma,P,S)$ is $NA(G)=(N,Sk\cup\Sigma,\{S\},\Delta)$ with 
$\Delta:=\{\sigma(U_1,\ldots,U_m)\to X\mid (X\to U_1\ldots U_m)\in P\}.$
Conversely, the $\epsilon$-free CFG corresponding to an NFTA $\cA=(\cQ,Sk\cup\Sigma,\cQ_{\mathtt{f}},\Delta)$ over the skeletal set $Sk\cup\Sigma$ is $G(\cA)=(\cQ\cup\{S\},\Sigma,P,S)$ where $S$ is a fresh symbol and 
$P:=
\{q\to q_1\ldots q_m\mid (\sigma(q_1,\ldots,q_m)\to q)\in\Delta\}\cup
\{S\to q_1\ldots q_m\mid (\sigma(q_1,\ldots,q_m)\to q)\in\Delta\text{ with }q\in \cQ_{\mathtt{f}}\}.
$
These  constructs are dual to each other, in the following sense:
\begin{enumerate}
\item[$(A_1)$] If $G$ is an $\epsilon$-free CFG then $\cL(NA(G))=K(D(G))$. \hfill \cite[Prop. 3.4]{Sak:1990}
\item[$(A_2)$] If $\cA=(\cQ,Sk\cup\Sigma,\cQ_{\mathtt{f}},\Delta)$ is an NFTA for the skeletal set $Sk\cup\Sigma$ then $K(D(G(\cA)))=\cL(\cA)$. That is, the set of structural descriptions of $G(\cA)$ coincides with the set of trees accepted by $\cA$. \hfill\cite[Prop. 3.6]{Sak:1990}
\end{enumerate}
We recall the following well-known results: 
every NFTA is equivalent to an DFTA~\cite{LJ:1978}, and every two minimal DFTAs are isomorphic~\cite{Brainerd:68}.

\subsection*{Cover tree automata}

\begin{definition}[determinstic DCTA]
Let $\ell\in\mathbb{N}^+$ and $A$ be a tree language over ranked alphabet $\cF$.
A {\em deterministic cover tree automaton} (DCTA) of  $A$ with respect to $\ell$ is a DFTA $\cA$ over a skeletal set $Sk\cup\cF_0$ such that $\cL(\cA)_{[\ell]}=K(A)_{[\ell]}$. 
\end{definition}
The correspondence between tree automata and $\epsilon$-free CFGs is carried over to a  correspondence between cover tree automata and cover CFGs. More precisely, it can be shown that if $G_U$ is an $\epsilon$-free CFG, then a DFTA $\cA$ is a DCTA of $K(D(G_U))$ w.r.t. $\ell$ if and only if $G(\cA)$ is a cover CFG of $G_U$ w.r.t. $\ell$.

\section{Learning context-free grammars}

In \cite{Sak:1990}, Sakakibara's  assumes  a {\em learner} eager to learn a CFG which is structurally equivalent with the CFG $G_U$ of an unknown context-free language $U\subseteq \Sigma^*$ by asking  questions to a {\em teacher}. We assume that the learner and the teacher share the skeletal set $Sk\cup\Sigma$ for the structural descriptions in $K(D(G_U))$. 
The learner can pose the following types of queries:
\begin{enumerate}
\item {\em Structural membership queries}: the learner asks if some $s\in\cT(Sk\cup\Sigma)$ is in $K(D(G_U))$. The answer is {\em yes} if so, and {\em no} otherwise.
\item {\em Structural equivalence queries}: The learner proposes a CFG $G'$ and asks whether $G'$ is structurally equivalent to $G_U$. 
If the answer is {\em yes}, the  process stops with the learned answer $G$. Otherwise, the teacher provides a counterexample $s$ from the symmetric set difference  $K(D(G')) \mathop{\triangle} K(D(G_U))$.
\end{enumerate}
This learning protocol is based on what is called {\em minimal adequate teacher} in~\cite{Angluin:87}.
Ultimately, the learner constructs a minimal DFTA $\cA$ of $K(D(G_U))$ from which it can infer immediately the CFG  $G'=G(\cA)$ which is structurally equivalent to $G_U$, that is, $K(D(G'))=K(D(G_U))$.  In order to understand how $\cA$ gets constructed, we shall introduce a few auxiliary notions. 

For any subset $S$ of $\cT(Sk\cup\Sigma)$, we define the sets 
\begin{align*}
\sigma_\hole\pair{S}&:=\bigcup_{m\in ar(\sigma)}\bigcup_{i=1}^m\{\sigma(s_1,\ldots,s_{m})[\hole]_i\mid s_1,\ldots,s_{m}\in S\cup\Sigma\},\\
X(S)&:=\{C_1[s]\mid C_1\in \sigma_\hole\pair{S},s\in S\cup\Sigma\}\setminus S.
\end{align*}
Note that $\sigma_\hole\pair{S}=\{C\in\cC(Sk\cup\Sigma)\setminus\{\hole\}\mid C|_p\in S\cup\Sigma\cup\{\hole\}$ for all $p\in Pos(C)\cap\mathbb{N}\}.$ 
\begin{definition}
A subset $E$ of $\cC(Sk\cup\Sigma)$ is {\bf $\bullet$-prefix closed} with respect to a set $S\subseteq \cT(Sk\cup\Sigma)$ if $C\in E\setminus\{\bullet\}$ implies the existence of $C'\in E$ and $C_1\in \sigma_\hole\pair{S}$ such that $C=C'[C_1].$ If $E\subseteq \cC(Sk\cup\Sigma)$ and $S\subseteq \cT(Sk\cup\Sigma)$ then $E[S]$ denotes the set of structural descriptions defined by
$E[S]=\{C[s]\mid C\in E,s\in S\}.$

We say that $S\subseteq\cT(Sk\cup\Sigma)$ is {\bf subterm closed} if $\depth(s)\geq 1$ for all $s\in S$, and $s'\in S$ whenever $s'$ is a subterm of some $s\in S$ with $\depth(s')\geq 1$.
\end{definition}

An {\em observation table} for $K(D(G_U))$, denoted by $(S,E,T)$, is a tabular representation of the finitary function $T:E[S\cup X(S)]\to\{0,1\}$
defined by $T(t):=1$ if $t\in K(D(G_U))$, and 0 otherwise,
where $S$ is a finite nonempty subterm closed subset $S$ of $\cT(Sk\cup\Sigma)$, and $E$ is a finite nonempty subset of $\cC(Sk\cup\Sigma)$ which is $\bullet$-prefix closed with respect to $S$.
Such an observation table is visualised as a matrix with rows labeled by elements from $S\cup X(S)$, columns labeled by elements from $E$, and the entry for row of $s$ and column of $C$ equal to $T(C[s])$. If we fix a listing $\pair{C_1,\ldots,C_r}$ of all elements of $E$, then the row of values of some $s\in S\cup X(S)$ corresponds to the vector $row(s)=\pair{T(C_1[s]),\ldots,T(C_r[s])}$. In fact, for every such $s$, $row(s)$ is a finitary representation of the function $f_s:E\to \{0,1\}$ defined by $f_s(C)=T(C[s])$. 

\comment{\begin{figure}[ht]
$$\begin{array}{|c|ccc|} \hline
T&\ldots&C\in E&\ldots \\ \hline
\vdots&&\vdots& \\
s\in S& \ldots&T(C[s])&\ldots \\ 
\vdots& &\vdots&\\ \hline\hline
\vdots& &\vdots&  \\ 
x\in X(S)&\ldots&T(C[x]) &\ldots\\
\vdots& &\vdots&  \\ \hline
\end{array}$$
\caption{Tabular representation of an observation table $(S,E,T)$.}
\label{opa}
\end{figure}}

The observation table $(S,E,T)$ is {\em closed} if every  $row(x)$ with $x\in X(S)$ is identical to some $row(s)$ of $s\in S$. It is 
{\em consistent} if whenever $s_1,s_2\in S$ such that $row(s_1)=row(s_2)$, we have $row(C_1[s_1])=row(C_1[s_2])$
for all $C_1\in\sigma_\hole\pair{S}.$

The DFTA {\em corresponding to a closed and consistent observation table} $(S,E,T)$ is $\cA(S,E,T)=(\cQ,Sk\cup\Sigma,\cQ_{\mathtt{f}},\delta)$ where
$\cQ:=\{row(s)\mid s\in S\}$, $\cQ_{\mathtt{f}}:=\{row(s)\mid s\in S\text{ and }T(s)=1\}$, and $\delta$ is uniquely defined by
$$\delta_m(\sigma)(q_1,\ldots,q_m):=row(\sigma(r_1,\ldots,r_m))\quad\text{for all }m\in ar(\sigma),$$
where $r_i:=a$ if $q_i=a\in\Sigma$, and $r_i:=s_i$ if $q_i=row(s_i)\in\cQ$.
\vskip .4em\noindent
It is easy to check that, under these assumptions, $\cA(S,E,T)$ is well-defined, and that $\delta^*(s)=row(s)$. Furthermore, Sakakibara proved that the following properties hold whenever $(S,E,T)$ is a closed and consistent observation table:
\begin{enumerate}
\item $\cA(S,E,T)$ is consistent with $T$, that is, for all $s\in S\cup X(S)$ and $C\in E$ we have $\delta^*(C[s])\in \cQ_{\mathtt{f}}$ iff $T(C[s])=1.$\hfill \cite[Lemma 4.2]{Sak:1990}
\item If $\cA(S,E,T)=(\cQ,Sk\cup\Sigma,\delta,\cQ_{\mathtt{f}})$ has $n$ states, and $\cA'=(\cQ',Sk\cup \Sigma,\delta',\cQ_{\mathtt{f}}')$ is any DFTA consistent with $T$ that has $n$ or fewer states, then $\cA'$ is isomorphic to $\cA(S,E,T)$.\hfill \mbox{\cite[Lemma 4.3]{Sak:1990}}
\end{enumerate}
\subsection*{The $LA$ algorithm}
In this subsection we briefly recall Sakakibara's algorithm LA whose pseudocode is given in Appendix \ref{LA}.
$LA$ extends the observation table whenever one of the following situations occurs: the table is not consistent, the table is not closed, or the table is both consistent and closed but the CFG corresponding to the resulting automaton $\cA(S,E,T)$  is not structurally equivalent to $G_U$ (in which case a counterexample is produced). The first two situations trigger an extension of the observation table with one distinct row.
From properties $(A_1)$ and $(A_2)$, if $n$ is the number of states of the minimum bottom-up tree automaton for the structural descriptions of  $G_U$, then the number of unsuccessful consistency and closedness checks during the whole run of this algorithm is at most $n-1$. 
For each counterexample of size at most $m$ returned by a structural equivalence query, at most $m$ subtrees are added to $S$.  Since the algorithm encounters at most $n$ counterexamples, the total number of elements in $S$ cannot exceed $n+m\cdot n$, thus $LA$ must terminate. It also follows that the number of elements of the domain $E[S\cup X(S)]$ of the function $T$ is at most
$(n+m\cdot n+\,(l+m\cdot n+k)^d)\cdot n=O(m^d\cdot n^{d+1})$,
where $l$ is the number of distinct ranks of $\sigma\in Sk$, and $d$ is the maximum rank of a symbol in $Sk$.  A careful analysis of $LA$ reveals that its time complexity is indeed bounded by a polynomial in $m$ and $n$ \cite[Thm. 5,3]{Sak:1990}.

\section{Learning cover context-free grammars}
\label{sect4}
We assume we are given a teacher who knows an $\epsilon$-free CFG $G_U$ for a language $U\subseteq\Sigma^*$, and a learner who  knows the skeletal set $Sk\cup\Sigma$ for $K(D(G_U))$.  
The teacher and learner both know a positive integer $\ell$, and the learner is interested to learn a cover CFG $G'$ of $G_U$ w.r.t. $\ell$ or, equivalently, a cover DCTA of $K(D(G_U))$ w.r.t. $\ell$.
The learner is allowed to pose the following types of questions:
\begin{enumerate}
\item {\em Structural membership queries}: the learner asks if some $s\in\cT(Sk\cup\Sigma)_{[\ell]}$ is in $K(D(G_U))$. The answer is {\em yes} if so, and {\em no} otherwise.
\item {\em Structural equivalence queries}: The learner proposes a CFG $G'$, and asks if $G'$ is a cover CFG of $G_U$ w.r.t. $\ell$.  If the answer is {yes}, the  process stops with the learned answer $G'$. Otherwise, the teacher provides a counterexample from the set  $(K(D(G_U))_{[\ell]}-K(D(G')))\cup (K(D(G'))_{[\ell]}- K(D(G_U))).$
\end{enumerate}
We will describe an algorithm $LA^\ell$ that learns a cover CFG of $G_U$ with respect to $\ell$ in time that is polynomial in the number of states of a minimal DCTA of the rational tree language $K(D(G_U)).$
\subsection{The observation table}
$LA^\ell$ is a generalisation of the learning algorithm $L^\ell$ proposed by Ipate \cite{Ipate:2012}. Ipate's algorithm is designed to learn a minimal finite cover automaton of an unknown finite language of words in polynomial time, using membership queries and language equivalence queries that refer to words and languages of words with  length at most $\ell$. Similarly,  $LA^\ell$ is designed to learn a minimal DCTA $\cA'$ for $K(D(G_U))$ with respect to $\ell$ by maintaining an observation table $(S,E,T,\ell)$ for $K(D(G_U))$ which differs from the observation table of $LA$ in the following respects:
\begin{enumerate}
\item $S$ is a finite nonempty subterm closed subset of $\cT(Sk\cup\Sigma)_{[\ell]}$. 
\item  $E$ is a finite nonempty  subset of $\cC(Sk\cup\Sigma)_{\pair{\ell-1}}\cap \cC(Sk\cup\Sigma)_{[\ell]}$ which is $\bullet$-prefix closed with respect to $S$.
\item $T:E[S\cup X(S)_{[\ell]}]\to\{1,0,-1\}$ is defined by 
$$T(t):=\left\{\begin{array}{rcl}
1&\quad&\text{if }t\in K(D(G_U))_{[\ell]},\\
0&&\text{if }t\in \cT(Sk\cup\Sigma)_{[\ell]}\setminus K(D(G_U)),\\
-1&&\text{if }t\not\in \cT(Sk\cup\Sigma)_{[\ell ]}.
\end{array}\right.$$
\end{enumerate}
In a tabular representation, the observation table $(S,E,T,\ell)$ is a two-dimensional matrix with rows labeled by elements from $S\cup X(S)_{[\ell]}$, columns labeled by elements from $E$, and the entry corresponding to the row of $t$ and column of $C$ equal to $T(C[t]).$
If we fix a listing $\pair{C_1,\ldots,C_k}$ of all elements from $E$, then the row of $t$ in the observation table is described by the vector $\pair{T(C_1[t]),\ldots,T(C_k[t])}$ of values from $\{-1, 0, 1\}.$ The rows of an observation table are used to identify the states a a minimal DCTA for $K(D(G_U))$ with respect to $\ell$. But, like Ipate \cite{Ipate:2012}, we do not compare rows by equality but by a similarity relation. 

\subsection{The similarity relation}
This time, the rows in the observation table correspond to terms from $S\cup X(S)_{[\ell]}$, and the comparison of rows should take into account only terms of depth at most $\ell$. For this purpose, we define a relation $\sim_k$ of {\em $k$-similarity}, which is a  generalisation to terms of Ipate's relation of $k$-similarity on strings \cite{Ipate:2012}.
\begin{definition}[$k$-similarity] 
For $1\leq k\leq \ell$ we define the relation $\sim_k$ on the elements of  the set $S\cup X(S)$ of an observation table $(S,E,T,\ell)$ as follows:
\begin{itemize}
\item[] $s\sim_k t$ if, for every $C\in E_{\pair{k-\max\{\depth(s),\depth(t)\}}}$, $T(C[s])=T(C[t]).$ 
\end{itemize}
When the relation $\sim_k$ does not hold between two terms $s,t\in S\cup X(S)$, we write $s\nsim_k t$ and say that $s$ and $t$ are {\em $k$-dissimilar}. When $k=\ell$  we simply say that $s$ and $t$ are {\em similar} or {\em dissimilar} and write $s\sim t$ or $s\nsim t$, respectively.

We say that a context $C$ {\bf $\ell$-distinguishes} $s_1$ and $s_2$, where $s_1,s_2\in S$, if $C\in E_{\pair{\ell-\max\{\depth(s_1),\depth(s_2)\}}}$ and $T(C[s_1])\neq T(C[s_2]).$
\end{definition}
Note that only the contexts $C\in E_{\pair{k-\max\{\depth(s),\depth(t)\}}}$ with $\depth(C)\leq \ell$ are relevant to check whether $s\sim_k t$, because if $\depth(C)>\ell$ then $\depth(C[s])>\ell$ and $\depth(C[t])>\ell$, and therefore $T(C[s])=-1=T(C[t]).$ Also, if $t\in S\cup X(S)$ with $\depth(t)>\ell$ then it must be the case that $t\in X(S)$, and then $t\sim_k s$ for all $s\in S\cup X(S)$ and $1\leq k\leq \ell$ because $E_{\pair{k-\max{\depth(t),\depth(s)\}}}}=\emptyset.$

The relation of $k$-similarity is obviously reflexive and symmetric, but  not transitive. The following example illustrates this fact.
\begin{example}
Let $\Sigma=\{a,b\}$, $k=1$, $\ell=2$, $S=\{\sigma(a),\sigma(b),\sigma(\sigma(a),b)\},$ $ E=\{\hole,\sigma(\hole,b)\},$ $ t_1=\sigma(a)$, $t_2=\sigma(\sigma(a),b)$, $t_3=\sigma(b)$, and
$$G_U=(\{\tS,\tA\},\{a,b\},\{\tS\to a,\tS\to b,\tS\to \tA b, \tA\to a,\tA\to \tA b\},\tS).$$ $S$ is a nonempty subterm closed subset of $\cT(Sk\cup\Sigma)_{[\ell]}$, and $E$ is a nonempty subset of $\cC(Sk\cup\Sigma)_{\pair{\ell-1}}$ which is $\bullet$-prefix closed with respect to $S$. We have $K(D(G_U))_{[\ell]}=\{t_1,t_2,t_3\}$, 
$t_1\sim_\ell t_2$ because $E_{\pair{\ell-\max\{\depth(t_1),\depth(t_2)\}}}=\{\hole\}$ and $T(\hole[t_1])=1=T(\hole[t_2]),$ and
$t_2\sim_\ell t_3$ because  $E_{\pair{\ell-\max\{\depth(t_2),\depth(t_3)\}}}=\{\hole\}$ and $T(\hole[t_2])=1=T(\hole[t_3]),$
However, $t_1\nsim_\ell t_3$ because $C=\sigma(\hole,b)\in E_{\pair{1}}=E_{\pair{\ell-\max\{\depth(t_1),\depth(t_3)\}}}$ and $T(C[t_1])=T(\sigma(\sigma(a),b))=T(t_2)=1$, but $T(C[t_3])=T(\sigma(\sigma(b),b))=0.$\qed
\end{example}
Still, $k$-similarity has a useful property, captured in the following lemma.
\begin{lemma}
\label{lema1}
Let $(S,E,T,\ell)$ be an observation table. If $s,t,x\in S\cup X(S)$ such that $\depth(x)\leq\max\{\depth(s),\depth(t)\}$, then $s\sim_k t$ whenever $s\sim_k x$ and $x\sim_k t.$
\end{lemma}

In addition, we will also assume given a total order $\prec$ on the alphabet $\Sigma$, and the following total orders induced by $\prec$ on $\cT(Sk\cup\Sigma)$ and $\cC(Sk\cup\Sigma).$
\begin{definition}
\label{def1}
The total order $\prec_{\mathtt{T}}$ on $\cT(Sk\cup\Sigma)$ induced by a total order $\prec$ on $\Sigma$ is defined as follows: $s\prec_{\mathtt{T}} t$ if either (a) $\depth(s)<\depth(t)$, or (b) $\depth(s)=\depth(t)$ and
\begin{enumerate}
\item $s,t\in\Sigma$ and $s\prec t$, or else
\item $s\in\Sigma$ and $t\not\in\Sigma$, or else
\item $s=\sigma(s_1,\ldots,s_m)$, $t=\sigma(t_1,\ldots,t_n)$ and there exists $1\leq k\leq \min(m,n)$ such that $s_k\prec_{\mathtt{T}} t_k$ and $s_i=t_i$ for all $1\leq i<k$, or else
\item $s=\sigma(s_1,\ldots,s_m)$ and $t=\sigma(t_1,\ldots,t_n)$, $m<n$, and $s_i=t_i$ for $1\leq i\leq m$.
\end{enumerate}
The total order $\prec_{\mathtt{C}}$ on $\cC(Sk\cup\Sigma)$ induced by a total order $\prec$ on $\Sigma$ is defined as follows:
$C_1\prec_{\mathtt{C}} C_2$ if either (a) $\depth_\hole(C_1)<\depth_\hole(C_2)$, or (b) $\depth_\hole(C_1)=\depth_\hole(C_2)$ and $C_1\prec_{\mathtt{T}}C_2$ where $C_1,C_2$ are interpreted as terms over the signature with $\Sigma$ extended with the constant $\hole$ such that $\hole\prec a$ for all $a\in \Sigma.$
\end{definition}
\begin{definition}[representative]
Let $(S,E,T,\ell)$ be an observation table and $x\in S\cup X(S)$. We say $x$ \emph{has a representative} in $S$ if $\{s\in S\mid s\sim x\}\neq \emptyset.$ If so, the \emph{representative} of  $x$ is $\tr(x):=\min_{\prec_{\mathtt{T}}} \{s\in S\mid x\sim s\}.$
\end{definition}
We will show later that the construction an observation table $(S,E,T,\ell)$ is instrumental to the construction of a cover tree automaton,  and the states of the automaton correspond to representatives of the elements from $S\cup X(S)$.
Note that, if $(S,E,T,\ell)$ is an observation table and $x\in S\cup X(S)$ has $\depth(x)>\ell$ then $x\in X(S)$ and $x\sim s$ for all $s\in S$. Then $s\prec_{\mathtt{T}} x$ because $\depth(s)\leq \ell<\depth(x)$ for all $s\in S$. Thus $x$ has a representative in $S$, and $\tr(x)=\min_{\prec_{\mathtt{T}}} S.$ For this reason, only the rows for elements $x\in S\cup X(S)_{[\ell]}$ are kept in an observation table. 
\subsection{Consistency and closedness}
\comment{
Using this definition, we choose the representative $\tr(t)$ of $t\in S\cup X(S)$ in a closed observation table $(S,E,T,\ell)$ to be the minimum element of the set $\{s\in S\mid s\sim t\}.$ Note that $\{s\in S\mid s\sim t\}\neq \emptyset$ when $(S,E,T)$ is closed,  thus $\tr(t)$ is well defined for all $t\in S\cup X(S)$. Also, if $t\in X(S)$ has $\depth(t)>\ell$ then $s\sim t$ for all $s\in S$, and thus $\tr(t)$ is the minimum element of $S$. For this reason, we we do not need to keep a row for elements  $t\in X(S)\setminus \cT(Sk\cup\Sigma)_{[\ell]}$ in the observation table because $\tr(t)$ is known a priori.
}

The   consistency and closedness of an observation table are defined as follows.
\begin{definition}[Consistency]
An observation table $(S,E,T,\ell)$ is consistent if, for every $k\in\{1,\ldots, \ell\}$, $s_1,s_2\in S$, and $C_1\in\sigma_\hole\pair{S}$, the following implication holds:
If $s_1\sim_k s_2$ then $C_1[s_1]\sim_k C_1[s_2]$. 
\end{definition}
The following lemma captures a useful property of consistent observation tables.
\begin{lemma}
\label{lele}
Let $(S,E,T,\ell)$ be a consistent observation table. Let $m\in ar(\sigma)$, $1\leq k\leq \ell$, and $s_1,\ldots,s_m,t_1,\ldots,t_m\in S\cup\Sigma$ such that, for all $1\leq i\leq m$, either 
$s_i=t_i\in\Sigma$, or 
$s_i,t_i\in S$, $s_i\sim_k t_i,$ and $\depth(s_i)\leq\depth(t_i)$,
and $s=\sigma(s_1,\ldots,s_m)$, $t=\sigma(t_1,\ldots,t_m).$ Then $s\sim_k t.$
\end{lemma}

\begin{definition}[Closedness]
An observation table $(S,E,T,\ell)$ is closed if, for all $x\in X(S)$, there exists $s\in S$ with $\depth(s)\leq \depth(x)$ such that $x\sim s.$
\end{definition}
The next five lemmata capture  important properties of closed observation tables:
\begin{lemma}
\label{lema2}
If $(S,E,T,\ell)$ is closed then every $x\in S\cup X(S)$ has a representative, and $\depth(\tr(x))\leq \depth(x)$.
\end{lemma}
\begin{lemma}
\label{lema4}
If $(S,E,T,\ell)$ is closed, $r_1,r_2\in\{\tr(x)\mid x\in S\cup X(S)\}$, and $r_1\sim r_2$ then $r_1=r_2.$ 
\end{lemma}
\begin{lemma}
\label{lema5}
If  $(S,E,T,\ell)$ is closed and $r\in\{\tr(x)\mid x\in S\cup X(S)\}$, then $\tr(r)=r.$
\end{lemma}
\begin{proof}
Let $r_1=\tr(r).$ Then $r_1\sim r$ and $r_1,r\in \{\tr(x)\mid x\in S\cup X(S)\}.$ By Lemma~\ref{lema4}, $r=r_1.$
\qed
\end{proof}
\begin{lemma}
\label{lema6}
If $(S,E,T,\ell)$ is closed, then for every $x\in S\cup X(S)$ and $C_1\in\sigma_\hole\pair{S}$, there exists $s\in S$ such that $\tr(C_1[\tr(x)])= \tr(s)$.
\end{lemma}

\begin{lemma}
\label{lem6}
Let $(S,E,T,\ell)$ be closed, $r\in \{\tr(x)\mid x\in S\cup X(S)\}$, $C_1\in\sigma_\hole\pair{S}$, and $s\in S$.  If $C_1[s]\sim r$ then $\depth(r)\leq \depth(C_1[s]).$
\end{lemma}
\subsection*{The automaton $\cA(\mathbb{T})$}
Like  $L^\ell$, our algorithm relies on the construction of a consistent and closed observation table of the unknown context-free grammar. The table is used to build an automaton which, in the end, turns out to be a minimal DCTA for the structural descriptions of the unknown grammar.
\begin{definition}
Suppose $\mathbb{T}=(S,E,T,\ell)$ is a closed and consistent observation table. The automaton corresponding to this table, denoted by $\cA(\mathbb{T})$, is the DFTA $(\cQ,Sk\cup\Sigma,\cQ_{\mathtt{f}},\delta)$ where
$\cQ:=\{\tr(s)\mid s\in S\}$, $\cQ_{\mathtt{f}}=\{q\in \cQ\mid T(q)=1\}$, and $\delta$ is uniquely defined by
$\delta_m(\sigma)(q_1,\ldots,q_m):=\tr(\sigma(q_1,\ldots,q_m))$ for all $m\in ar(\sigma)$.
\end{definition}
The transition function $\delta$ is well defined because, for all $m\in ar(\sigma)$ and $q_1,\ldots,q_m$ from $\cQ$, 
$C_1:=\sigma(\hole,q_2,\ldots,q_m)\in \sigma_\hole\pair{S}$, thus $\sigma(q_1,\ldots,q_m)=C_1[q_1]\in S\cup X(S)$ and $\tr(C_1[q_1])=\tr(s)$ for some $s\in S$, by Lemma \ref{lema6}. Hence, $\tr(\sigma(q_1,\ldots,q_m))\in \cQ$. 
Also, the set $\cQ_{\mathtt{f}}$ can be read off directly from the observation table because 
$\hole\in E$ (since $E$ is $\hole$-prefix closed), thus $q=\hole[q]\in E[(S\cup X(S)_{[\ell]}]$ for all $q\in \cQ$,
and we can read off from the observation table all $q\in\cQ$ with $T(q)=1.$
\vskip .5em
In the rest of this subsection we assume  that $\mathbb{T}=(S,E,T,\ell)$ is  closed and consistent, and  $\delta$ is the transition function of the corresponding DFTA $\cA(\mathbb{T}).$
\begin{lemma}
\label{lmm}
$\delta^*(x)\sim x$ and $\depth(\delta^*(x))\leq\depth(x)$ for every $x\in S\cup X(S)$.
\end{lemma}

\begin{corollary}
\label{cor1}
$\delta^*(x)=x$ for all $x\in\{\tr(s)\mid s\in S\cup X(S)\}.$
\end{corollary}
\begin{proof}
By Lemma \ref{lmm}, $x\sim \delta^*(x).$ Since both $\delta^*(x)$ and $x$ belong to the set of representatives $\{\tr(s)\mid s\in X\cup X(S)\}$,  $x=\delta^*(x)$ by Lemma \ref{lema4}.\qed
\end{proof}
The following theorem shows that the DFTA of a closed and consistent observation table is consistent with the function $T$ on terms with depth at most $\ell$.
\begin{theorem}
\label{thm1}
Let $\mathbb{T}=(S,E,T,\ell)$ be a closed and consistent observation table. For every $s\in S\cup X(S)$ and $C\in E$ such that $\depth(C[s])\leq\ell$ we have $\delta^*(C[s])\in \cQ_{\mathtt{f}}$ if and only if $T(C[s])=1.$
\end{theorem}

\begin{theorem}
\label{oops}
Let $\mathbb{T}=(S,E,T,\ell)$ be a closed and consistent observation table, and $N$ be the number of states of $\cA(\mathbb{T}).$ If $\cA'$ is any other DFTA with $N$ or fewer states, that is consistent with $T$ on terms with depth at most $\ell$, then $\cA'$ has exactly $N$ states and $\cL(\cA(\mathbb{T}))_{[\ell]}=\cL(\cA')_{[\ell]}$. 
\end{theorem}

\begin{corollary}
\label{cor2}
Let $\cA$ be the automaton corresponding to a closed and consistent observation table $(S,E,T,\ell)$ of the skeletons of a CFG $G_U$ of an unknown language $U$, and $N$ be its number of states. Let $n$ be the number of states of a minimal DCTA of $K(D(G_U))$ with respect to $\ell$. If $N\geq n$ then $N=n$ and $\cA$ is a minimal DCTA of $K(D(G_U))$ with respect to $\ell$.
\end{corollary}
\subsection*{The $LA^\ell$ algorithm}
The algorithm $LA^\ell$ extends the observation table $\mathbb{T}=(S,E,T,\ell)$ whenever one of the following situations occurs: the table is not consistent, the table is not closed, or the table is both consistent and closed but the resulting automaton $\cA(\mathbb{T})$  is not a cover tree automaton of $K(D(G_U))$ with respect to $\ell$.

The pseudocode of the algorithm is shown below. 
\begin{tabbing}
ask if $(\{\tS\},\Sigma,\emptyset,\tS)$ is a cover CFG of $G_U$ w.r.t. $\ell$\\
{\bf if} answer is {\em yes} {\bf then} halt and output the CFG $(\{\tS\},\Sigma,\emptyset,\tS)$\\
{\bf if}\=\ answer is {\em no} with counterexample $t$ {\bf then}\\
\>set $S:=\{s\mid s$ is a subterm of $t$ with depth at least $1\}$ and $E=\{\hole\}$\\
\>construct \=the table $\mathbb{T}=(S,E,T,\ell)$ using structural membership queries\\
\>{\bf rep}\={\bf{eat}}\\
\>\ \ {\bf repeat}\\
\>\>/* check consistency */\\
\>\>{\bf for}\=\ every $C\in E,$ in increasing order of $i=\depth_\hole(C)$ {\bf do}\\
\>\>\>search \=for $s_1,s_2\in S$ with $\depth(s_1),\depth(s_2)\leq \ell-i-1$ and $C_1\in\sigma_\hole\pair{S}$\\
\>\>\>\>such that \=$C[C_1[s_1]]), C[C_1[s_2]]\in\cT(Sk\cup\Sigma)_{[\ell]}$,\\
\>\>\>\>\>$s_1\sim_k s_2$ where $k=\max\{\depth(s_1),\depth(s_2)\}+i+1$,\\
\>\>\>\>\>and $T(C[C_1[s_1]])\neq T(C[C_1[s_2]])$\\
\>\>\>{\bf if} \=found {\bf then}\\
\>\>\>\>add $C[C_1]$ to $E$\\
\>\>\>\>extend $T$ to $E[S\cup X(S)_{[\ell]}]$ using structural membership queries\\
\>\>/* check closedness */\\
\>\>$new\_row\_added:=\mathtt{false}$\\
\>\>{\bf repeat} for every $s\in S$, in increasing order of $\depth(s)$\\
\>\>\>search for $C_1\in\sigma_\hole\pair{S}$ such that \=$C_1[s]\nsim t$ for all $t\in S_{[\depth(C_1[s])]}$ \\
\>\>\>{\bf if} \=found {\bf then}\\
\>\>\>\>add $C_1[s]$ to $S$\\
\>\>\>\>extend $T$ to $E[S\cup X(S)_{[\ell]}]$ using structural membership queries\\
\>\>\>\>$new\_row\_added:=\mathtt{true}$\\
\>\>{\bf until} $new\_row\_added=\mathtt{true}$ or all elements of $S$ have been processed\\
\>\ \ {\bf until} $new\_row\_added=\mathtt{false}$\\
\>\ \ /* $\mathbb{T}$ is now closed and consistent */ \\
\>\ \ make the query whether $G(\cA(\mathbb{T}))$ is a cover CFG of $G_U$ w.r.t. $\ell$\\
\>\ \ {\bf if}\=\ the reply is {\em no} with a counterexample $t$ {\bf then}\\
\>\>add to $S$ \=all subterms of $t$, including $t$, with depth at least 1, \\
\>\>\>in the increasing order given by $\prec_{\mathtt{T}}$\\
\>\>extend $T$ to $E[S\cup X(S)_{[\ell]}]$ using structural membership queries\\
\>{\bf until} \=the reply is {\em yes} to the query if $G(\cA(\mathbb{T}))$ is a cover CFG of $G_U$ w.r.t. $\ell$\\
\>halt and output $G(\cA(\mathbb{T}))$.
\end{tabbing}

Consistency is checked by searching for $C\in E$ and $C_1\in\sigma_\hole\pair{S}$ such that $C[C_1]$ will $\ell$-distinguish  two terms  $s_1,s_2\in S$ not distinguished by any other context  $C'\in E$ with $\depth_\hole(C')\leq\depth_\hole(C[C_1]).$ Whenever such a pair of contexts $(C,C_1)$ is found,  $C[C_1]$ is added to $E$. Note that 
$C[C_1]\in\cC(Sk\cup\Sigma)_{\pair{\ell-1}}\cap\cC(Sk\cup\Sigma)_{[\ell]}$ because only such contexts can distinguish terms from $S$, and the addition of $C[C_1]$ to $E$ yields a $\hole$-prefix closed subset of $\cC(Sk\cup\Sigma)_{\pair{\ell-1}}\cap \cC(Sk\cup\Sigma)_{[\ell]}$.

The search of such a pair of contexts $(C,C_1)$ is repeated in increasing order of the hole depth of $C$, until all contexts from $E$ have been processed. Therefore, any context $C[C_1]$ with $C\in E$ and $C_1\in\sigma_\hole\pair{S}$ that was added to $E$ because of a failed consistency check will be processed itself in the same {\bf for} loop.

The algorithm checks closedness by searching for $s\in S$ and $C_1\in\sigma_\hole\pair{S}$ such that $C_1[s]\nsim t$ for all $t\in S$ for which $\depth(t)\leq\depth(C_1[s]).$ The search is performed in increasing order of the depth of $s$. If $s$ and $C_1$ are found,  $C_1[s]$ is added to the $S$ component of the observation table, and the algorithm checks again  consistency.  Note that adding $C_1[s]$ to $S$ yields a subterm closed subset of $\cT(Sk\cup\Sigma)_{[\ell]}$. Also, closedness checks are performed only on consistent observation tables.

When the observation table is both consistent and closed, the corresponding DFTA is constructed and it is checked whether the language accepted by the constructed automaton coincides with the set of skeletal descriptions of the unknown context-free grammar $G_U$ (this is called a {\em structural equivalence query}). If this query fails, a counterexample from $\cL(\cA(\mathbb{T}))_{[\ell]}\mathop{\triangle} K(D(G_U))_{[\ell]}$ is produced,  the component $S$ of the observation table is expanded to include $t$ and all its subterms with depth at least 1, and the consistency and closedness checks are performed once more. At the end of this step, the component $S$ of the observation table is subterm closed, and $E$ is unchanged, thus $\hole$-prefix closed.

Thus, at any time during the execution of algorithm $LA^\ell$, the defining properties of an observation table are preserved:
the component $S$ is a subterm closed subset of $\cT(Sk\cup\Sigma)_{[\ell]}$, and
the component $E$ is a $\hole$-prefix closed subset of $\cC(Sk\cup\Sigma)_{\pair{\ell-1}}\cap \cC(Sk\cup\Sigma)_{[\ell]}$.

\section{Algorithm analysis}
\label{sect5}
We notice that the number of states of the DFTA constructed by algorithm $LA^\ell$ will always increase between two successive structural equivalence queries. When this number of states reaches the number of states of a minimal DCTA of $K(D(G_U))$, the constructed DFTA is actually a minimal DCTA of $K(D(G_U))$ (Corollary \ref{cor2}) and the algorithm terminates.

From now on we assume implicitly that $n$ is the number of states of a minimal DCTA of $K(D(G_U))$ with respect to $\ell$, and that $\mathbb{T}(\ty)$ is the observation table $(S^{\ty},E^{\ty},T,\ell)$ before execution step $\ty$ of the algorithm. 
By Corollary \ref{cor2}, $\cQ^{\ty}$ will always have between 1 and $n$ elements. Note that the representative of an element $s\in S$ in $\cQ^\ty$ is a notion that depends on the observation table $\mathbb{T}(\ty).$ Therefore, we will use the notation $\tr_\ty(s)$ to refer to the representative of $s\in S^\ty$ in the observation table $\mathbb{T}(\ty).$ With this notation, $\cQ^\ty=\{\tr_\ty(s)\mid s\in S^\ty\}.$

Note that the execution of algorithm $LA^\ell$ is a sequence of steps characterised by the detection of three kinds of failure:  closedness,  consistency, and structural equivalence query. The $\ty$-th execution step is
\begin{enumerate}
\item a failed closedness check when the algorithm finds $C_1\in\sigma_\hole\pair{S^{\ty}}$ and $s\in S^{\ty}$ such that $C_1[s]\nsim t$ for all $t\in S^{\ty}$ with $\depth(t)\leq\depth(C_1[s])$,
\item a failed consistency check when the algorithm finds $C\in E^{\ty}$ with $\depth_\hole(C)=i$, $s_1,s_2\in S^{\ty}$ with $\depth(s_1),\depth(s_2)\leq \ell-i-1$, and $C_1\in\sigma_\hole\pair{S^{\ty}}$, such that  $C[C_1[s_1]],$ $C[C_1[s_2]]\in\cT(Sk\cup\Sigma)_{[\ell]}$, $s_1\sim_k s_2$ where $k=\max\{\depth(s_1),\depth(s_2)\}+i+1$, and $T(C[C_1[s_1]])\neq T(C[C_1[s_2]])$,
\item a failed structural equivalence query when the observation table $\mathbb{T}(\ty)$ is closed and consistent, and the learning algorithm receives from the teacher a counterexample $t\in\cT(Sk\cup\Sigma)_{[\ell]}$ as answer to the structural equivalence query with the grammar $G(\cA(S^{\ty},E^{\ty},T,\ell)).$
\end{enumerate}
In the following subsections we perform a complexity analysis of the algorithm by identifying upper bound estimates to the computations due to failed  consistency checks, failed closeness checks, and failed structural equivalence queries. 

\comment{
First, we define inductively an enumeration ordering of the elements of $\cQ^{\ty}$.
If the enumeration of the elements  of $\cQ^{\ty}$ is  $r_1,\ldots,r_m$ then the enumeration of the elements of $\cQ^{\ty+1}$ is defined as follows:
\begin{enumerate}
\item If the $\ty$-th execution step is a failed closedness check which introduces a term $r'$ from $X(S^{\ty})_{[\ell]}$ into $\cQ^{\ty}$ (that is, $r'\in \cQ^{\ty+1}\setminus\cQ^\ty$) then:
\begin{enumerate}
\item If there is no $r\in \cQ^{\ty}$ with $r'\sim r$ then $\cQ^{\ty+1}$ has $m+1$ elements, and the enumeration of its elements is $r_1,\ldots,r_m,r'.$
 \item Otherwise, there is an $r_j\in\cQ^{\ty}$ such that $r'\sim r_j$ and $\depth(r')<\depth(r_j).$ Note that there is no $r_k\neq r_j$ such that $r'\sim r_k$ because Lemma \ref{lema1} would imply the contradiction $r_j\sim r_k.$ In this case, $\cQ^{\ty+1}$ has $m$ elements, and the ordering of its elements is $r_1,\ldots,r_{j-1},r',r_{j+1},\ldots,r_m.$
 \end{enumerate}
 \item If the $\ty$-th execution step is a failed consistency check which introduces a context $C\in \cC(Sk\cup\Sigma)_{\pair{\ell-1}}\cap\cC(Sk\cup\Sigma)_{[\ell]}$ into $E^{\ty+1}$ then $\cQ^{\ty}\subseteq S^{\ty+1}$ and, in the  table $\mathbb{T}(\ty+1)$, $\tr(q)=q$ for all $q\in\cQ^{\ty}$. Therefore $\cQ^{\ty +1}=\cQ^{\ty}\cup \cQ'$ where $\cQ'$ is the set of representatives newly created by the introduction of context $C$ into $E^{\ty+1}.$ 
In this case, the enumeration of the elements of $\cQ^{\ty+1}$ is $$r_1,\ldots,r_m,r'_1,\ldots,r'_p$$ where 
$r'_1,\ldots,r'_p$ is the enumeration of the elements of $\cQ'$  given by $\prec_{\mathtt{T}}$.
 \item If the $\ty$-th execution step is a failed structural membership query with a counterexample $t\in \cT(Sk\cup\Sigma)_{[\ell]}$ then $S^{\ty+1}=S^{\ty}\cup M$ where $M$ is the set of subterms of $t$, including $t$. Let $\{t_1,\ldots,t_u\}=M\setminus S^{\ty}$, where $t_1\prec_{\mathtt{T}}\ldots\prec_{\mathtt{T}} t_u$. Then we can trace the computation of the observation table $\mathbb{T}(\ty+1)$ from the observation table $\mathbb{T}(\ty)$ via the computation of the sequence of $u$ intermediary observation tables $\mathbb{T}'_j=(S_j,E,T,\ell)$ with $S_j:=S^\ty\cup\{t_1,\ldots,t_j\}$ for all $0\leq j\leq u$. Note that $\mathbb{T}'_0=\mathbb{T}(\ty)$ is a closed observation table, and $\mathbb{T}(\ty+1)=\mathbb{T}'_u$. Let $\cQ'_j:=\{\tr(s)\mid s\in S_j\}$ in the table $\mathbb{T}'_j$ for all $0\leq j\leq u$. Let $1\leq i\leq u$, and suppose the enumeration of elements of $\cQ'_{i-1}$ is $r'_1,\ldots,r'_k$.
 For the computation of $\cQ'_i$  from $\cQ'_{i-1}$ we distinguish the following situations: 
 \begin{enumerate}
 \item If there is no $r\in \cQ'_{i-1}$ with $t_i\sim r$ in $\mathbb{T}'_{i-1}$ then $\cQ'_i$ has $k+1$ elements, and the enumeration of its elements is $r'_1,\ldots,r'_k,t_i.$ This situation is similar to the situation 1.(a) for a failed closedness check.
 \item If there is $r\in\cQ'_{i-1}$ with $t_i\sim r$ in $\mathbb{T}'_{i-1}$ and $r\prec_{\mathtt{T}}t_i$ then $\depth(r)\leq\depth(t_i)$, $\cQ'_{i}=\cQ'_{i-1}$, and $\cQ'_{i}$ has the same enumeration of elements as $\cQ'_{i-1}$.
\item If there is $r_j\in\cQ'_{i-1}$ with $r_j\sim t_i$  in $\mathbb{T}'_{i-1}$ and $t_i\prec_{\mathtt{T}}r_j$ then $\depth(t_i)\leq\depth(r_j)$ and there is no  $r'_j\neq r'_{j'}\in \cQ'_{i-1}$ with $r_{j'}\sim t_i$ in $\mathbb{T}'_{i-1}$ because this would imply $r'_{j'}\sim r'_j$ by Lemma \ref{lema1}, and thus the contradiction $r'_{j'}=r'_j$ by Lemma \ref{lema4}.  In this situation, $\cQ'_i=(\cQ'_{i-1}\setminus\{r'_j\})\cup\{t_i\}$ has $k$ elements, and the enumeration of its elements is $r'_1,\ldots,r'_{j-1},t_i,r'_{j+1},\ldots,r'_k$.
This situation is similar to the situation 1.(b) for a failed closedness check.
 \end{enumerate}
\end{enumerate}
From now on we denote by $\tty_\ty(j)$ the $j$-th element of $\cQ^\ty$ in this enumeration ordering. 
A useful property of this definition is indicated in the following lemma.
\begin{lemma}
For any  $r=\tty_\ty(i)\in\cQ^\ty$, if $r'$ is a subterm of $r$ and $r'\not\in\Sigma$, then the representative of $r'$ in the observation table $\mathbb{T}(\ty)$ is $\tty_\ty(j)$ for some $j\leq i$.
\end{lemma}
\begin{proof}
This property is obviously preserved by the execution steps $\ty$ corresponding to cases 1.(a), 1, 3.(a), and 3.(b). In cases 1.(b) and 3.(c), the execution step replaces a term $\tty_\ty(j)$ with a new term $r'\prec_{\mathtt{T}}\tty_\ty(j)$ after all proper subterms of $r'$ are guaranteed to have representatives in the $S$ component of the observation table.
\todo
\qed
\end{proof}
}
\subsection{Failed closedness checks}
We recall that the $\ty$-th execution step is a failed closedness check  if the algorithm finds a context $C_1\in\sigma_\hole\pair{S}$ and a term $s\in S^\ty$ such that $C_1[s]\nsim t$ for all $t\in S^\ty$ with $\depth(t)\leq \depth(C_1[s])$. 
We will show that the number of failed closedness checks performed by  algorithm $LA^\ell$ has an upper bound which is a polynomial in $n$. To prove this fact, we will rely on the following auxiliary notions:
\begin{itemize}
\item For $r,r'\in\cQ^\ty$, we define $r\prec^\ty_{\mathtt{T}} r'$ if either $\depth(r)<\depth(r')$ or $\depth(r)=\depth(r')$ and there exists $\ty'<\ty$ such that $r\in\cQ^{\ty'}$ but $r'\not\in\cQ^{\ty'}$ (that is, $r$ became a representative in the observation table before $r'$).
\item To every set of representatives $\cQ^\ty=\{r_1,\ldots,r_m\}$ with $r_1\prec^\ty_{\mathtt{T}}\ldots\prec^\ty_{\mathtt{T}} r_m$
we associate the tuple $\tpl(\cQ^\ty):=(d_1,\ldots,d_n)\in\{1,\ldots,\ell+1\}^n$ where  
$d_i:=
\depth(r_i)$ if $1\leq i\leq m,$ and $d_i:=\ell+1$ if $m < i\leq n.$
\item We consider the following  partial order on $\mathbb{N}^n$: $(x_1,\ldots,x_n) < (x'_1,\ldots,x'_n)$ iff there exists $i\in\{1,\ldots,n\}$ such that $x_i< x'_i$ and $x_j\leq x'_j$ for all $1\leq j\leq n$.
\item We denote by $\tty_\ty(i)$ the $i$-th component of  $\cQ^\ty$ in the order given by $\prec^{\ty}_{\mathtt{T}}$. 
\end{itemize}
\comment{
\begin{lemma}
\label{lmng}
Let $\mathbf{x}=(x_1,\ldots,x_n),\ \mathbf{y}= (y_1,\ldots,y_n)\in\mathbb{N}^n$ such that 
$x_1\leq \ldots\leq x_n$ and $y_1\leq \ldots \leq y_n.$ Then
$\{x_1,\ldots,x_n\}\prec_{\ms}\{y_1,\ldots,y_n\}.$ Then $\mathbf{x}<\mathbf{y}$.
\end{lemma}
\begin{proof}
Obvious.\qed
\end{proof}
\comment
It is easy to see that the following inequalities hold for all execution steps $\ty$: 
$$\tpl(\cQ^\ty)\geq\tpl(\cQ^{\ty+1})\text{ and }(1,\ldots,1)\leq \tpl(\cQ^\ty)\leq (\ell+1,\ldots,\ell+1).$$
}

\begin{lemma}
\label{lema8}
Suppose $s$ has been introduced in $S^{\ty+1}$ as a result of a failed closedness check. There exists  $p\in Pos(s)$ such that $\|p\|=\depth(s)$ and for every prefix $p'$ of $p$ different from $p$, $\depth(\tr_{\ty+1}(s|_{p'}))=\depth(s|_{p'})$.
\end{lemma}

\begin{corollary}
\label{cor3}
Whenever the $\ty$-th execution step is a failed closedness check, the term introduced in $S^{\ty+1}$ is in $\cQ^{\ty+1}\setminus\cQ^{\ty}$ and its depth is at most $j$, where $j$ is the position in $\cQ^{\ty+1}$ of the newly introduced element  according to  ordering $\prec^\ty_{\mathtt{T}}.$
\end{corollary}

\begin{corollary}
\label{mmon}
$\depth(s)\leq n$ for all $s\in S^\ty$ which was introduced in the table by a failed closedness check.
\end{corollary}
\begin{proof}
$\depth(s)\leq j$ by Cor. \ref{cor3}, and $j\leq n$ because $|\cQ^\ty|\leq n$ for all $\ty$. Thus $\depth(s)\leq n.$ 
\end{proof}
\begin{lemma}
\label{cor4}
Let $j$ be the position of the element  introduced in $\cQ^{\ty+1}$ by a failed closedness check. Then 
$\tpl(\cQ^{\ty+1})<\tpl(\cQ^\ty)$ and $\depth(\tty_{\ty+1}(j))<\depth(\tty_{\ty}(j)).$
\end{lemma}

\begin{theorem}
\label{tyury}
The number of failed closedness checks performed during the entire run of $LA^\ell$ is at most $n(n+1)/2$.
\end{theorem}
\subsection{Failed consistency checks}
The $\ty$-th execution step is a failed consistency check if the algorithm finds  $C\in E^{\ty}$ with $\depth_\hole(C)=i$, $s_1,s_2\in S^{\ty}$ with $\depth(s_1),\depth(s_2)\leq \ell-i-1$, and $C_1\in\sigma_\hole\pair{S^{\ty}}$, such that  $C[C_1[s_1]],$ $C[C_1[s_2]]\in\cT(Sk\cup\Sigma)_{[\ell]}$, $s_1\sim_k s_2$ where $k=\max\{\depth(s_1),\depth(s_2)\}+i+1$, and $T(C[C_1[s_1]])\neq T(C[C_1[s_2]])$. In this case, the context $C[C_1]$ is newly introduced in the component $E^{\ty+1}$ of the observation table $\mathbb{T}({\ty+1})$.

We will show that the number of failed consistency checks performed by the learning algorithm $LA^\ell$ has an upper bound which is a polynomial in $n$. 
 To prove this fact, we  rely on the following auxiliary notions:
\begin{itemize}
\item For $C,C'\in E^\ty$, we define $C\prec^\ty_{\mathtt{C}} C'$ if either $\depth_\hole(C)<\depth_\hole(C')$ or $\depth_\hole(C)=\depth_\hole(C')$ and there exists $\ty'<\ty$ such that $C\in E^{\ty'}$ but $C'\not\in E^{\ty'}$ (that is, $C$ became an experiment in the observation table before $C'$).
\item We define 
$\delta_{\ty}(s_1,s_2):=\min_{\prec_{\mathtt{C}}}\{C\in E^{\ty}\mid C\text{ $\ell$-distinguishes $s_1$ and $s_2$}\}$
for every $s_1,s_2\in S^{\ty}$ such that $s_1\nsim s_2$.
\item A nonempty subset $U$ of $E^{\ty}$ induces a partition of a subset $R$ of $S^{\ty}$ into equivalence classes $Q_1,\ldots,Q_m$ if the following conditions are satisfied:
\begin{enumerate}
\item $\bigcup_{j=1}^m Q_j=R$ and $Q_i\cap Q_j=\emptyset$ whenever $1\leq i\neq j\leq m$,
\item Whenever $1\leq i\neq j\leq m$, $s_1\in Q_i,$ and $s_2\in Q_j$, there exists $C\in U$ that $\ell$-distinguishes $s_1$ and $s_2$.
\item Whenever $s_1,s_2\in Q_j$ for some $1\leq j\leq m$, there is no $C\in U$ that $\ell$-distinguishes $s_1$ and $s_2$.
\end{enumerate}
\end{itemize}
Let $\cE^{\ty}:=\{\delta_{\ty}(s_1,s_2)\mid s_1,s_2\in S^{\ty},s_1\nsim s_2\}$. 
%
Since $\sim$ is not an equivalence, not every subset of $E^{\ty}$ induces a partition of $S^\ty$ into equivalence classes. However, the next lemma shows that $\cE^{\ty}$  induces a partition of $\cQ^{\ty}$ into at least $|\cE^{\ty}|$ classes.
\begin{theorem}
\label{conscheck}
If $\cE^{\ty}=\{C_1,\ldots,C_k\}$ with $C_1\prec_{\mathtt{C}}\ldots\prec_{\mathtt{C}} C_k$ then, for every $1\leq i\leq k$,  $\{C_1,\ldots,C_i\}$ induces a partition of $\cQ^\ty$ into at least $i$ classes.
\end{theorem}
 \begin{corollary}
 \label{cor5}
 For any $\ty$, $\cE^\ty$ has at most $n$ elements.
 \end{corollary}
We will compute an upper bound on the number of failed consistency checks by examining the evolution of $\cE^{\ty}$ during the execution of  $LA^\ell.$ Initially, $\cE^0=\{\hole\}.$ 
 \begin{lemma}
\label{tmy4}
At any time during the execution of the algorithm, if $\cQ^\ty$ has $i\geq 2$ elements, then the hole depth of any context in $E^\ty$ is less than or equal to $i-2$.
\end{lemma}

Let $\cE^\ty=\{C'_1,\ldots,C'_k\}$ before some execution step $\ty$ of the algorithm $LA^\ell$, where $C'_1\prec_{\mathtt{C}}\ldots\prec_{\mathtt{C}}C'_k$. Then $k\leq n$ by Cor. \ref{cor5}. We associate to every such $\cE^\ty$ the $n$-tuple $\tpl(\cQ^\ty)=(y_1,\ldots,y_n)\in\{0,1,\ldots,n-1\}^n$, where, for every $1\leq j\leq n$, $y_j$ is defined as follows:\par
- If $\cQ^\ty$ has at least $j+1$ elements then, if $i$ is the minimum integer such that $\{C'_1,\ldots,C'_i\}$ partitions $\cQ^\ty$ into at least $j+1$ classes then $y_j=\depth_\hole(C'_j)$. Since every $\{C'_1,\ldots,C'_i\}$ partitions $\cQ^\ty$ into at least $i$ classes (by Lemma \ref{conscheck}) and we assume that $\cE^\ty=\{C'_1,\ldots,C'_k\}$ partitions $\cQ^\ty$ into $|\cQ^\ty|\geq j+1$ classes, 
we conclude that such $i$ exists.\par
- otherwise $y_j=n-1$.\vskip .4em
For  $1\leq j\leq n$ we denote the $j$-th component of $\tpl(\cE^\ty)$ by $\xpm_\ty(j).$
Note that, for all $1\leq i\leq k$, 
$\depth_\hole(C'_i)\leq |\cQ^\ty|-2$ by Theorem \ref{tmy4},  and $|\cQ^\ty|\leq n$, hence $\depth_\hole(C'_i)\leq n-2$. Therefore, we can always distinguish the components $y_i$ of $\tpl(\cQ^\ty)$ that correspond to the defining case (1) from those 
in case (2).
\begin{lemma}
\label{lups}
$\xpm_\ty(j)\leq j-1$ whenever $2\leq j\leq n$ and $\xpm_\ty(j)\neq n-1.$
\end{lemma}
\begin{theorem}
\label{inconsistentthm}
If $\cQ^\ty$ has at least 2 elements then the number of failed consistency checks over the entire run of $LA^\ell$ is at most $n(n-1)/2.$
\end{theorem}

\subsection{Failed structural equivalence queries} 
Every failed structural equivalence query yields a counterexample which increases the number of representatives in $\cQ^\ty$. Thus 
\begin{theorem}
\label{failstruct}
The number of failed structural equivalence queries is at most $n$.
\end{theorem}

\subsection{Space and time complexity}
We are ready now to express the space and time complexity of $LA^\ell$ in terms of the following parameters:\par
- $n$ = the number of states of a minimal DFCA for the language of structural descriptions of the unknown grammar with respect to $\ell$,\par
- $m$ = the maximum size of a counterexample returned by a failed structural equivalence query,\par
- $p$ = the cardinality of the alphabet $\Sigma$ of terminal symbols, and\par
- $d$ = the maximum rank (or arity) of the symbol $\sigma\in Sk.$ \vskip .3em
First, we determine the space needed by the observation table. The number of elements in $S^\ty$ is initially 0 (i.e., $|S^0|=0$) and is increased either by a failed closedness check or by a failed structural membership query. By Theorem \ref{tyury}, the number of failed closedness checks is at most $n(n+1)/2$, and each of them adds one element to $S$. By Theorem  \ref{failstruct}, the number of failed structural equivalence queries is at most $n$. A failed structural equivalence query which produces a counterexample $t$ with $\sz(t)\leq m$, adds at most $m$ terms to $S^\ty$. Thus, $|S^\ty|\leq n(n+1)/2+n\,m=O(m\,n+n^2)$ and $|S^\ty\cup \Sigma|=O(m\,n+n^2+p)$,  therefore $|\sigma_\hole\pair{S^\ty}|\leq \sum_{j=0}^{d-1}(j+1)\,|S^\ty\cup\Sigma|^j=O((d+1)\,(m\,n+n^2+p)^d)$ and $|X(S)|\leq \sum_{j=1}^d|S^\ty\cup\Sigma|^j=O(d\,(m\,n+n^2+p)^d).$ Thus $S^\ty\cup X(S^\ty)_{[\ell]}$ has $O(d\,(m\,n+n^2+p)^d)$ elements.
By Theorem \ref{inconsistentthm}, there may be at most $n(n-1)/2$ failed consistency checks, and each of them adds a context to $E^\ty$. Thus $E^\ty$ has $O(n^2)$ elements and $E^\ty[S^\ty\cup X(S^\ty)_{[\ell]}]$ has $O(n^2d\,(m\,n+n^2+p)^d)$ elements. By Lemma \ref{lups}, $\depth_\hole(C)\leq n-1$ for all $C\in E^\ty$. We also know that, if $s\in S^\ty$, then $\depth(s)\leq m$ if it originates from a failed structural equivalence query, and $\depth(s)\leq n$ if it originates from a failed closedness check (by Cor. \ref{mmon}). Therefore $\depth(s)\leq\max(m,n)$ for all $s\in S^\ty$, and thus $\depth(x)\leq 1+\max(m,n)\leq 1+m+n$ for all $x\in S^\ty\cup X(S^\ty)$ and $\depth(t)\leq m+2\,n$ for all $t\in E^\ty[S^\ty\cup X(S^\ty)_{[\ell]}].$ Since the number of positions of such a term $t$ is $\sum_{j=0}^{m+2\,n}d^j=O((m+2\,n+1)d^{m+2\,n})$, we conclude that the total space occupied by an observation table at any time is $O\left(n^2(m\,n+n^2+p)^d(m+2\,n+1)d^{m+2\,n+1}\right).$

Next, we examine the time complexity of the algorithm by looking at the time needed to perform each kind of operation. 

Since the consistency checks of the observation table are performed in  a {\bf for} loop which checks the result produced by $s_1\sim_k s_2$ (where $s_1,s_2\in S^\ty$) in increasing order of $k$, the result produced by $s_1\sim_k s_2$ can be reused in checking $s_1\sim_{k+1} s_2$ and so the corresponding elements in the rows of $s_1$ and $s_2$ are compared only once. Thus, the total time needed to check if the observation table is consistent involves at most $(|S^\ty|\cdot (|S^\ty|-1)/2)\cdot |E^\ty|\cdot (1+|\sigma_\hole\pair{S^\ty}|)$, comparisons.
As $\sigma_\hole\pair{S^\ty}$ has  $O(d\,(m\,n+n^2+p)^d)$ elements, a consistency check of the  table takes $O((m\,n+n^2)^2 n^2 d\,(m\,n+n^2+p)^d)=O(n^2 d\,(m\,n+n^2+p)^{d+2})$ time.
As there are at most $(n\,(n+1)/2+1)\,(n+1)=O(n^3)$ consistency checks, the total time needed to check if the  table is consistent is $O(n^5 d\,(m\,n+n^2+p)^{d+2}).$

Checking if the observation table is closed takes at most $|S^\ty|^2\cdot |\sigma_\hole\pair{S^\ty}|\cdot |E^\ty|$ time, which is 
$O((m\,n+n^2)^2 d\,(m\,n+n^2+p)^d n^2)=O(n^2d\,(m\,n+n^2+p)^{d+2}).$

Extending an observation table $\mathbb{T}(\ty)$ with a new element in $S^{\ty+1}$ requires the addition of $\sum_{k=2}^d (2^{k-1}-1)=2^d-d-1$ contexts to $\sigma_\hole\pair{S^{\ty+1}}\setminus \sigma_\hole\pair{S^\ty}$, thus the addition of at most $2^d-d$ new rows for the new elements of $S^{\ty+1}\cup X(S^{\ty+1})$ in the observation table $\mathbb{T}(\ty+1).$ This extension requires 
at most
$(2^d-d)\cdot|E^\ty|\cdot(1+|\sigma_\hole\pair{S^\ty}|)=O(n^2d\,(2^d-d)\,(m\,n+n^2+p)^{d})$ membership queries. The number of elements added to $S^\ty$ as a result of a failed structural equivalence query is at most $m.$ As there will be at most $n$ failed structural equivalence queries and at most $n(n+1)/2$ failed closedness checks, the maximum number of elements added to $S^\ty$ is $n(n+1)/2+m\,n=O(m\,n+n^2).$ Thus the total time spent on inserting new elements in the $S$-component of the observation table is $O(n^2d\,(2^d-d)\,(m\,n+n^2)(m\,n+n^2+p)^d)$. 
 Adding a context to $E^\ty$ requires at most $|S^\ty+X(S^\ty)_{[\ell]}|=O(d\,(m\,n+n^2+p)^d)$ membership  queries. These additions are performed only by failed consistency checks, and there are at most $n(n-1)/2$ of them. Thus, the total time spent to insert new contexts in the $E$-component of the observation table is $O(n^2d\,(m\,n+n^2+p)^d).$
We conclude that the total time spent to add elements to the components $S$ and $E$ of the observation table is $O(n^2d\,(2^d-d)\,(m\,n+n^2)(m\,n+n^2+p)^d)$, which is polynomial.
 
The identification of the representative $\tr_\ty(s)$ for every $s\in S^\ty$ can be done by performing $((|S^\ty|)(S^\ty-1)/2)\,|E^\ty|=O((m\,n+n^2)^2n^2)$ comparisons.

Thus, all DFCAs $\cA(\mathbb{T}(\ty))$ corresponding to consistent and closed observation tables $\mathbb{T}(\ty)$ can be constructed in time  polynomial in $m$ and $n$. Since the algorithm encounters at most $n$ consistent and closed observation tables, the total running time of the algorithm is polynomial in $m$ and $n$.
\section{Conclusions and acknowledgments}
\label{sect6}
We have presented an algorithm, called $LA^\ell$, for learning cover context-free grammars from structural descriptions of languages of interest. $LA^\ell$ is an adaptation of Sakakibara's algorithm $LA$ for learning context-free grammars from structural descriptions, by following a methodology similar to the design of Ipate's algorithm $L^\ell$ as a nontrivial adaptation of Angluin's algorithm $L^*$. Like $L^*$, our algorithm synthesizes  a minimal deterministic cover automaton consistent with an observation table maintained via a learning protocol based on what is called in the literature a ``minimally adequate teacher'' \cite{Angluin:87}. And again, like algorithm $L^*$, our algorithm is guaranteed to synthesize the desired automaton in time  polynomial in $n$ and $m$, where $n$ is its number of states and $m$ is the maximum size of a counterexample to a  structural equivalence query. As the size of a minimal finite cover automaton is usually much smaller than that of a minimal automaton that accepts that language, the algorithm $LA^\ell$ is a  better choice than algorithm $LA$ for applications where we are interested only in an accurate characterisation of the structural descriptions with depth at most $\ell$.

This work has been supported by CNCS IDEI Grant PN-II-ID-PCE-2011-3-0981 ``Structure and computational difficulty in combinatorial
optimization: an interdisciplinary approach.''

\bibliographystyle{abbrv}
\bibliography{bibliography}
\newpage
\section*{Appendix}
\section{Pseudocode of algorithm $LA$}
\label{LA}
\begin{tabbing}
Set $S=\emptyset$ and $E=\{\bullet\}$\\
let $G':=(\{\tS\},\Sigma,\emptyset,\tS)$\\
check if $G'$ is structurally equivalent with $G_U$\\
{\bf if} answer is {\tt yes} {\bf then} halt and output $G'$\\
{\bf if}\=\ answer is {\tt no} with counterexample $t$ {\bf then}\\
\>add $t$ and all its subterms with depth at least 1 to $S$\\
\>construct the observation table $(S,E,T)$ using structural membership queries\\
\>{\bf rep}\={\bf{eat}}\\
\>\ \ {\bf while} $(S,E,T)$ is not closed or not consistent\\
\>\>{\bf if}\=\ $(S,E,T)$ is not consistent {\bf then}\\
\>\>\>fi\=nd $s_1,s_2\in S,C\in E$, and $C_1\in\sigma_\hole\pair{S}$ such that\\
\>\>\>\>$row(s_1)=row(s_2)$ and $T(C[C_1[s_1]])\neq T(C[C_1[s_2]])$\\
\>\>\>add $C[C_1]$ to $E$\\
\>\>\>extend $T$ to $E[S\cup X(S)]$ using structural membership queries\\
\>\>{\bf if} $(S,E,T)$ is not closed {\bf then}\\
\>\>\>find $s_1\in X(S)$ such that $row(s_1)\neq row(s)$ for all $s\in S$\\
\>\>\>add $s_1$ to $S$\\
\>\>\>extend $T$ to $E[S\cup X(S)]$ using structural membership queries\\
\>\ \ /* $(S,E,T)$ is now closed and consistent */ \\
\>\ \ let\=\ $G':=G(\cA(S,E,T))$\\
\>\ \ make the structural equivalence query between $G'$ and $G_U$\\
\>\ \ {\bf if}\=\ the reply is {\tt no} with a counterexample $t$ {\bf then}\\
\>\>add $t$ and all its subterms with depth at least 1 to $S$\\
\>\>extend $T$ to $E[S\cup X(S)]$ using structural membership queries\\
\>{\bf until} the reply is {\tt yes} to the structural equivalence query between $G'$ and $G_U$\\
\>halt and output $G'$.
\end{tabbing}

\section{Proof of Lemma \ref{lema1}}
Suppose  $s\sim_k x$ and $x\sim_k t.$ By definition of $\sim_k$, we have 
\begin{itemize}
\item[] $T(C[s])=T(C[x])$ for all $C\in E_{\pair{k-\max\{\depth(s),\depth(x)\}}}$, and
\item[] $T(C[x])=T(C[t])$ for all $C\in E_{\pair{k-\max\{\depth(x),\depth(t)\}}}$.
\end{itemize}
Let $m:=\max\{\depth(s),\depth(t)\}.$ Since $\depth(x)\leq m$, it follows that for every $C\in  E_{\pair{k-m}}$ we also have $C\in E_{\pair{k-\max\{\depth(s),\depth(x)\}}}$ and $C\in E_{\pair{k-\max\{\depth(x),\depth(t)\}}}$. Thus $T(C[s])=T(C[x])=T(C[t])$ for all $C\in  E_{\pair{k-m}}$. Hence $s\sim_k t.$
\section{Proof of Lemma \ref{lele}}
Let $I=\{i_1,\ldots,i_p\}=\{i\in\{1,\ldots,m\}\mid s_i,t_i\in S$ and $s_i\sim_k t_i\}$. 
If $I=\emptyset$ then $s=t$ and the result follows from the reflexivity of $\sim_k$. 
If $I\neq\emptyset$, let $x_0:=s$, and $x_j:=x_{j-1}[t_{i_j}]_{i_j}$ for $1\leq j\leq p$.
For all $1\leq j\leq p$ we have
$$\left.\begin{array}{r}
s_{i_j},t_{i_j}\in S\\ 
s_{i_j}\sim t_{i_j}\ (\text{by induction hypothesis})\\
x_{j-1}[\hole]_{i_j}\in\sigma_\hole\pair{S} 
\end{array}\right\}\Rightarrow x_{j-1}=x_{j-1}[s_{i_j}]_{i_j}\sim_k x_{j-1}[t_{i_j}]_{i_j}=x_j$$
because the observation table $(S,E,T,\ell)$ is consistent. Thus $x_0\sim_k x_1$, \ldots, $x_{p-1}\sim_k x_p$, and $\depth(x_0)\leq\depth(x_1)\leq\ldots\leq\depth(x_{p-1})\leq\depth(x_p).$ Repeated applications of Lemma \ref{lema1} yield $x_0\sim_k x_p.$ But $x_0=s$ and $x_p=t$, thus $s\sim_k t.$
\section{Proof of Lemma \ref{lema2}}
If $x\in S$ then $x$ has a representative since $\{s\in S\mid x\sim s\}\neq\emptyset$ and we can take $\tr(x)=\min_{\prec_{\mathtt{T}}}\{s\in S\mid x\sim s\}$. Then $\tr(x)\preceq_{\mathtt{T}} x$, which implies $\depth(\tr(x))\leq \depth(x).$ If $x\in X(S)$ then, since the observation table is closed, there exists $s\in S$ with $x\sim s$ and $\depth(s)\leq \depth(x).$ $x\sim s$ and $s\in S$ imply $\tr(x)\preceq_{\mathtt{T}} s$, hence $\depth(\tr(x))\leq \depth(s)$. Thus $\depth(\tr(x))\leq \depth(x)$ because $\depth(s)\leq \depth(x).$ 
\section{Proof of Lemma \ref{lema4}}
Suppose $r_1=\tr(x_1)$ and $r_2=\tr(x_2)$ for some $x_1,x_2\in S\cup X(S)$. By Lemma~\ref{lema2}, $r_1,r_2\in S$ and $\depth(r_1)\leq \depth(x_1)\leq\max\{\depth(x_1),\depth(r_2)\}.$ 
Since $x_1\sim r_1$ and $r_1\sim r_2$, Lemma~\ref{lema1} implies $x_1\sim r_2$, thus $r_2\in \{s\in S\mid x_1\sim s\}$ and $r_1=\min_{\preceq_{\mathtt{T}}}\{s\in S\mid x_1\sim s\}\preceq_{\mathtt{T}} r_2$. By a similar argument, we learn that $r_2\preceq_{\mathtt{T}} r_1$. From $r_1\preceq_{\mathtt{T}} r_2$ and $r_2\preceq_{\mathtt{T}} r_1$ we conclude that $r_1=r_2.$ 
\section{Proof of Lemma \ref{lema6}}
Let $x\in S\cup X(S)$ and $C_1\in\sigma_\hole\pair{S}$. The fact that $(S,E,T,\ell)$ is closed implies $\tr(x)\in S$, thus $C_1(\tr(x))\in S\cup X(S)$ and therefore $\tr(C_1[\tr(x)])\in S.$ We can choose $s:=\tr(C_1[\tr(x)])\in S$ for which $\tr(s)=s$, by Lemma \ref{lema5}. 
\section{Proof of Lemma \ref{lem6}}
We provide a proof by contradiction. Assume $\depth(C_1[s])>\depth(r).$
Since $C_1[s]\in S\cup X(S)$ and $(S,E,T,\ell)$ is closed, $\tr(C_1[s])\in S$, $\depth(\tr(C_1[s]))\leq \depth(C_1[s])$  (by Lemma \ref{lema2}),  $\tr(C_1[s])\sim C_1[s]$, and $C_1[s]\sim r$. Thus $\tr(C_1[s])\sim r$ by Lemma~\ref{lema1}.
Since $r,\tr(C_1[s])\in\{\tr(x)\mid x\in S\cup X(S)\}$, we have $r=\tr(C_1[s])$ by Lemma \ref{lema4}. Thus $\depth(r)=\depth(\tr(C_1[s]))\leq\depth(C_1[s])$, which yields a contradicton. 

\section{Proof of Lemma \ref{lmm}}
By induction on the depth of $x$. If $\depth(x)=1$ then 
$\delta^*(x)=\tr(x)\sim x$ and $\depth(\delta^*(x))=\depth(\tr(x)) \leq \depth(x)$ by Lemma \ref{lema2}. 

If $\depth(x)>1$ then $x=\sigma(s_1,\ldots,s_m)$ with $s_1,\ldots,s_m\in S\cup\Sigma$, and $\delta^*(x)=\tr(\sigma(q_1,\ldots,q_m)),$ where $q_i=\delta^*(s_i)$ for $1\leq i\leq m.$ Let $I:=\{i\mid s_i\not\in \Sigma\}.$ Then, by induction hypothesis for all $i\in I$, $q_i\sim s_i$ and $\depth(q_i)\leq\depth(s_i)$. 
Thus
$$\begin{array}{rl}
\left.\begin{array}{l}
\forall i\in I,\ \depth(q_i)\leq\depth(s_i)\\
\forall i\in \{1,\ldots,m\}\setminus I,\ q_i=s_i
\end{array}\right\}&
\Rightarrow \depth(\sigma(q_1,\ldots,q_m))\leq\depth(\sigma(s_1,\ldots,s_m))=\depth(x),\\
\delta^*(x)=\tr(\sigma(q_1,\ldots,q_m))&\Rightarrow\depth(\delta^*(x))\leq\depth(\sigma(q_1,\ldots,q_m)),\ \text{by Lemma 2.}
\end{array}$$
Hence $\depth(\delta^*(x))\leq \depth(x)$ follows from $\depth(\delta^*(x))\leq\depth(\sigma(q_1,\ldots,q_m))\leq\depth(x).$ 

To prove  $\delta^*(x)\sim x$, we 
\comment{reason as follows. Suppose $I=\{i_1,\ldots,i_p\}$, and 
let $x_0:=x$, and $x_j:=x_{j-1}[q_{i_j}]_{i_j}$ for $1\leq j\leq p$. 
For all $i_j\in I$ we have
$$\left.\begin{array}{r}
s_{i_j},q_{i_j}\in S\\ 
s_{i_j}\sim q_{i_j}\ (\text{by induction hypothesis})\\
x_j[\hole]_{i_j}\in\sigma_\hole\pair{S} 
\end{array}\right\}\Rightarrow x_{j-1}=x_j[s_{i_j}]_{i_j}\sim x_j[q_{i_j}]_{i_j}=x_j$$
because the observation table $(S,E,T,\ell)$ is consistent. Also, $\depth(x_j)\geq\depth(x_{j+1})$ whenever $0\leq j< p$ because, by induction hypothesis, $\depth(s_{i_j})\geq\depth(q_{i_j})$ for all $j\in I$. Hence $x=x_{i_0}\sim x_{i_p}=\sigma(q_1,\ldots,q_m)$ by repeated applications of Lemma~\ref{lema1}. }
notice that $x=\sigma(s_1,\ldots,s_m)\sim\sigma(q_1,\ldots,q_m)$ follows from Lemma \ref{lele}.
Thus
$$\left.\begin{array}{l}
\delta^*(x)=\tr(\sigma(q_1,\ldots,q_m))\sim \sigma(q_1,\ldots,q_m),\\
\sigma(q_1,\ldots,q_m)\sim x,\\
\depth(\sigma(q_1,\ldots,q_m))\leq \depth(x)\leq\max\{\depth(x),\depth(\delta^*(x))\}
\end{array}\right\}\Rightarrow \delta^*(x)\sim x\ \text{by Lemma \ref{lema1}}.$$
\section{Proof of Theorem \ref{thm1}}
Let $s\in S\cup X(S)$ and $C\in E$ such that $\depth(C[s])\leq\ell$. We proceed by induction on the hole depth of $C$. 
If $\depth_\hole(C)=0$ then $C=\hole$ and $C[s]=s$ has $\depth(s)\leq\ell.$  By Lemma \ref{lmm}, $\delta^*(s)\sim s$ and $\depth(\delta^*(s))\leq\depth(s).$ Thus, since $\hole\in E$ and $\depth(s)\leq\ell$, $T(\delta^*(s))=1$ if and only if $T(s)=1.$
By definition of $\cA(\mathbb{T})$, $\delta^*(s)\in \cQ_{\mathtt{f}}$ if and only if $T(\delta^*(s))=1$. Hence $\delta^*(s)\in \cQ_{\mathtt{f}}$ if and only if $T(s)=1.$

If $\depth_\hole(C)=m>0$ then $\depth(C[s])\leq \ell$ implies $m\leq \ell-\depth(s)$ and $C\in E_{\pair{m}}.$ Since $E$ is $\hole$-prefix closed, there exist $C'\in E_{\pair{m-1}}$ and $C_1\in\sigma_\hole\pair{S}$ such that $C=C'[C_1]\in E_{\pair{m}}.$ Let $t=\delta^*(C_1[s]).$ Then $\depth(t)\leq\depth(C_1[s])$ by Lemma \ref{lmm}, thus $\depth(C'[t])\leq \depth(C'[C_1[s]])=\depth(C[s])\leq \ell$, and we learn from the induction hypothesis for $C'$ that $\delta^*(C'[t])\in \cQ_{\mathtt{f}}$ if and only if $T(C'[t])=1.$ Since 
$$\left.\begin{array}{ll}
\delta^*(t)=t&\text{by Corollary \ref{cor1}}\\
t=\delta^*(C_1[s])\quad&\text{by definition}
\end{array}\right\}\Rightarrow \delta^*(C'[t])=\delta^*(C'[C_1[s]])=\delta^*(C[s]),$$
we have $\delta^*(C[s])\in\cQ_{\mathtt{f}}\Leftrightarrow \delta^*(C'[t])\in \cQ_{\mathtt{f}}\Leftrightarrow T(C'[t])=1.$
Therefore, it suffices to show that $T(C'[t])=1$ if and only if $T(C[s])=1$. By Lemma \ref{lmm}, $t\sim C_1[s]$ and $\depth(t)\leq\depth(C_1[s])$, thus $\depth(C'[t])\leq\depth(C'[C_1[s]])=\depth(C[s])\leq\ell.$
Hence,  since $C'\in E$ and $t\sim C_1[s]$,  $T(C'[t])=1$ if and only if $T(C'[C_1[s]])=1$.  
\section{Proof of Theorem~\ref{oops}}
Let $\cA'=(\cQ',Sk\cup\Sigma,\cQ'_{\mathtt{f}},\delta')$ and $f:\cQ\to\cQ'$ defined by $f(q)=\delta'^*(q)$ for all $q\in \cQ$. We show that $f$ is injective. If $q_1,q_2\in \cQ$ such that $q_1\neq q_2$ then $T(C[q_1])\neq T(C[q_2])$ for some $C\in E_{\pair{\ell-\max\{\depth(q_1),\depth(q_2)\}}}$. Since $\cA'$ is consistent with   $T$, exactly one of $\delta'^*(C[q_1])$ and $\delta'^*(C[q_2])$ is in $\cQ'_{\mathtt{f}}$. Hence $f(q_1)\neq f(q_2).$ 
Since $f$ is injective and $\cQ'$ has at most the same number of states as $\cQ$, $f$ is bijective. Thus $\cQ'=f(\cQ)=\{\delta'^*(q)\mid q\in\cQ\}.$

Next, we show that $\cQ'_{\mathtt{f}}=\{f(q)\mid q\in\cQ_{\mathtt{f}}\}.$ By Theorem \ref{thm1}, $\delta^*(q)\in \cQ_{\mathtt{f}}$ if and only if $T(q)=1.$ By Corollary \ref{cor1}, $\delta^*(q)=q$ for all $q\in\cQ$. Thus $\cQ_{\mathtt{f}}=\{q\in\cQ\mid T(q)=1\}.$ Similarly, since $\cA'$ is consistent with $T$, for every $q\in \cQ$, $T(q)=1$ if and only if $\delta'^*(q)\in\cQ'_{\mathtt{f}}.$ By definition, $\delta'^*(q)=f(q).$ Thus, $\cQ'_{\mathtt{f}}=\{\delta'^*(q)\mid q\in\cQ$ and $T(q)=1\}=\{f(q)\mid q\in\cQ_{\mathtt{f}}\}.$

We prove by induction on the depth of $x\in\cT(Sk\cup\Sigma)_{[\ell]}\setminus\Sigma$ that, if $\delta^*(x)=q$ and $\delta'^*(x)=f(q')$ then the following statements hold:
\begin{enumerate}
\item $\depth(q)\leq\depth(x)$,
\item  $\depth(q')\leq \depth(x)$,
\item if $m=\ell-\depth(x)+\max\{\depth(q),\depth(q')\}$, then $q\sim_m q'$.
\end{enumerate}
In the base case, $\depth(x)=1$
and $q=\tr(x)$. By Lemma \ref{lema2}, $q\sim x$ and $\depth(q)\leq \depth(x)$, thus $\depth(q)$ can only be $1$. $\delta'^*(x)=f(q')$ implies $\delta'^*(C[x])=\delta'^*(C[q'])$ for all $C\in E_{\pair{\ell-\max\{\depth(q'),\depth(x)\}}}$. Since $\cA'$ is consistent with $T$ on $\cT(Sk\cup\Sigma)_{[\ell]}$, this implies  $T(C[x])=1$ if and only if $T(C[q'])=1.$ Therefore  $x\sim q'$. From  $\tr(x)\sim x$, $x\sim q'$, and $\depth(x)=1\leq\max\{\tr(x),q'\}$, we learn by Lemma \ref{lema1} that $q'\sim \tr(x)=q.$  
Then $q=q'$ by Lemma \ref{lema4}, because $q,q'\in \cQ=\{\tr(s)\mid s\in S\}$ and $q\sim q'$. Thus $\depth(q')= \depth(q)\leq\depth(x).$ In this case, $m=\ell$ and statement 3 obviously holds because $\sim_\ell$
is reflexive.

In the induction step, we assume that all three statements hold for all terms $s\in \cT(Sk\cup\Sigma)_{[k]}\setminus\Sigma$ with $k\geq 1$. Let $x\in \cT(Sk\cup\Sigma)$ with $\depth(x)=k+1.$ Then $x=\sigma(x_1,\ldots,x_p)$ with $\depth(x_i)\leq k$ for $1\leq i\leq p.$ Let $I:=\{i\mid 1\leq i\leq p$ and $x_i\not\in\Sigma\}$,  and $q,q',q_i,q'_i\in\cQ$ such that $\delta^*(x)=q$,  $\delta'^*(x)=f(q')$, and $q_i=\delta^*(x_i)$, $\delta'^*(x_i)=f(q'_i)$ for all $i\in I$.

Let $y:=\sigma(s_1,\ldots,s_p)$ where $s_i:=x_i$ if $x_i\in \Sigma$ and $s_i:=q_i$ otherwise. Then $y\in S\cup X(S)$ and $q=\tr(y)$, thus $q\sim y$ and $\depth(q)\leq\depth(y)$ by Lemma \ref{lema2}. Also $\depth(y)\leq \depth(x)$ because $y=\sigma(s_1,\ldots,s_p),$ $x=\sigma(x_1,\ldots,x_p)$, and
\begin{itemize}
\item $\depth(s_i)=\depth(q_i)\leq \depth(x_i)$ for all $i\in I$, by induction hypothesis, 
\item $s_i=x_i$ for all $i\in\{1,\ldots,p\}\setminus I$, hence $\depth(s_i)=\depth(x_i)$ for all $i\in\{1,\ldots,p\}\setminus I$.
\end{itemize}
Thus $\depth(q)\leq\depth(x)$ follows from  $\depth(q)\leq\depth(y)$ and $\depth(y)\leq\depth(x).$

To show $\depth(q')\leq\depth(x)$, we reason as follows. $\delta'^*(q')=f(q')=\delta'^*(x)=\delta'_p(\sigma)(\delta'^*(x_1),\ldots,\delta'^*(x_p))$. Since $\delta'^*(x_i)=\delta'^*(q'_i)$ for all $i\in I$, and $\delta'^*(x_i)=x_i$ for all $i\in\{1,\ldots,p\}\setminus I$, we learn that $\delta'^*(q')=\delta'^*(z)$ where $z=\sigma(t_1,\ldots,t_p)$ with $t_i:=q'_i$ if $i\in I$ and $t_i:=x_i$ if $i\in\{1,\ldots,p\}\setminus I$.  Note that $z\in S\cup X(S)$ and $\delta'^*(C[q'])=\delta'^*(C[z])$ for all $C\in E_{\pair{\ell-\max\{\depth(q'),\depth(z)\}}}$. Thus $T(C[q'])=1$ if and only if $T(C[z])=1$ because $\cA'$ is consistent with $T$ on $\cT(Sk\cup\Sigma)_{[\ell]}$. Therefore $q'\sim z$, and since $z=C_1[q'_i]$ with $C_1=z[\hole]_i\in\sigma_\hole\pair{S}$ and $q'_i\in S$ for any $i\in I$, we can apply Lemma~\ref{lem6} to learn that $\depth(q')\leq\depth(z).$
Also
\begin{itemize}
\item for all $i\in I$, $\depth(t_i)=\depth(q'_i)\leq \depth(x_i)$ by induction hypothesis, and
\item for all $i\in \{1,\ldots,p\}\setminus I$, $t_i=x_i$, thus $\depth(t_i)=\depth(x_i)$,
\end{itemize}
therefore $\depth(z)=\depth(\sigma(t_1,\ldots,t_p))\leq\depth(\sigma(x_1,\ldots,x_p))=\depth(x).$ From $\depth(q')\leq\depth(z)$ and $\depth(z)\leq\depth(x)$ we learn $\depth(q')\leq\depth(x).$

Let $m=\ell-\depth(x)+\max\{\depth(q),\depth(q')\}$.
We prove $q\sim_m q'$ by contradiction. 
If $q\nsim_m q'$  there exists $C\in E_{\pair{\ell-\depth(x)}}$ such that $T(C[q])\neq T(C[q']).$ Then $q\sim y$ and $\depth(q)\leq\depth(y)\leq\depth(x)$, thus $\depth(C[q])\leq\depth(C[y])\leq\depth(C[x])\leq\ell$ and  $T(C[q])=T(C[y])$. 
Also, $q'\sim z$ and $\depth(q')\leq\depth(z)\leq\depth(x)$, thus $\depth(C[q'])\leq\depth(C[z])\leq\depth(C[x])\leq \ell$ and $T(C[q'])=T(C[z]).$ Thus $T(C[y])\neq T(C[z]).$ On the other hand, by induction hypothesis, $q_i\sim_{m_i} q'_i$ for all $i\in I$, where $m_i=\ell-\depth(x_i)+\max\{\depth(q_i),\depth(q'_i)\}.$
Let's assume $I=\{i_1,\ldots,i_r\}$,  $C_1:=y[\hole]_{i_1}$, and $C_{j+1}:=C_{j}[q'_{i_j}][\hole]_{i_{j+1}}$ for all $1\leq j< r.$ 
Then $C_j\in\sigma_\hole\pair{S}$ and $C_j[q_{i_j}]\sim_{m_{i_j}}C_j[q'_{i_j}]$ for all $1\leq j\leq r$, because the observation table is consistent. Therefore $T(C'_j[C_j[q_{i_j}]])=T(C'_j[C_j[q'_{i_j}]])$ whenever $1\leq j\leq r$ and  $C'_j\in E_{\pair{\ell-\depth(x_{i_j})-1}}$. Since $\depth(x)=1+\max\{\depth(x_{i_j})\mid 1\leq j\leq r\}$, we have
$C\in  E_{\pair{\ell-\depth(x_{i_j})-1}}$, thus $T(C[C_j[q_{i_j}]])=T(C[C_j[q'_{i_j}]])$ for all $1\leq k\leq r$.
 Note that
\begin{align*}
T(C[y])&=T(C[C_1[q_{i_1}]])=T(C[C_1[q'_{i_1}]])=T(C[C_2[q_{i_2}]])=T(C[C_2[q'_{i_2}]])=\ldots\\
&=T(C[C_r[q_{i_r}]])=T(C[C_r[q'_{i_r}]])=T(C[z])
\end{align*}
which yields a contradiction.

Finally, we prove that $\cL(\cA(\mathbb{T}))_{[\ell]}=\cL(\cA')_{[\ell]}$. Let $x\in \cT(Sk\cup\Sigma)_{[\ell]}$ and $q,q'\in \cQ$ such that $\delta^*(x)=q$ and $\delta'^*(x)=\delta'^*(q').$ Then $q\sim_m q'$ where $m=\ell-\depth(x)+\max\{\depth(q),\depth(q')\}.$ Since $\depth(x)\geq \max\{\depth(q),\depth(q')\}$ and $\hole\in E$, $T(q)=T(q')\in\{0,1\}.$  $\cA(\mathbb{T})$ is consistent with $T$ on $\cT(Sk\cup\Sigma)_{[\ell]}$ and $\delta^*(q)=q$, thus $q\in \cQ_{\mathtt{f}}$ if and only if $T(q)=1.$ $\cA'$ is also consistent with $T$ on $\cT(Sk\cup\Sigma)_{[\ell]}$, thus $\delta'^*(q')\in \cQ'_{\mathtt{f}}$ if and only if $T(q')=1.$ Since $T(q)=T(q')$, we have $q\in\cQ_{\mathtt{f}}$ if and only if $f(q')\in\cQ'_{\mathtt{f}}.$ Thus $\delta^*(x)\in\cQ_{\mathtt{f}}$ if and only if $\delta'^*(x)\in\cQ_{\mathtt{f}}'$. That is,
$x\in \cL(\cA(\mathbb{T}))_{[\ell]}\text{ if and only if }x\in \cL(\cA')_{[\ell]}.$
\section{Proof of Corollary \ref{cor2}}
Let $\cA'$ be a minimal DCTA of $K(D(G_U))$ with respect to $\ell$. Then $\cA'$ is consistent with  $T$ on $\cT(Sk\cup\Sigma)_{[\ell]}$ and has $n$ states. Since $n\leq N$, by Theorem~\ref{oops}, $n=N$ and $\cL(\cA)_{[\ell]}=\cL(\cA')_{[\ell]}=K(D(G_U))_{[\ell]}.$ Thus $\cA$ is a minimal DCTA of $K(D(G_U))$  with respect to $\ell$. 

\section{Proof of Lemma \ref{lema8}}
We prove by induction on $i$ that for every $i\in\{0,\ldots,\depth(s)-1\}$ there exists a sequence of positions $p_0<p_1<\ldots <p_i$ from $Pos(s)$ such that, for all $0\leq j\leq i$, the following statements hold:
\begin{enumerate}
\item[](L1): $\|p_j\|=j$ and $\depth(s|_{p_j})=\depth(s)-j$,  
\item[](L2): $\depth(\tr_{\ty+1}(s|_{p_j}))=\depth(s|_{p_j})$.
\end{enumerate}

For $i=0$ we reason as follows: Since $s$ has been introduced in $S^{\ty+1}$ as a result of a failed closedness check, $s\nsim t$ for all $t\in S^{\ty}$ with $\depth(t)\leq\depth(s).$ Then $s$ becomes a new element of the set $\cQ^{\ty+1}$, $\tr_{\ty+1}(s)=s$ and, if we choose $p_0=\epsilon$, the sequence of positions $p_0$ fulfils  requirements  (L1) and (L2). 

For the inductive step, assume the condition holds for $0\leq i<\depth(s)-1$, that is, there exists a sequence of positions 
$p_0<\ldots<p_i$ from $Pos(s)$ which fulfils requirements (L1) and (L2) for all $0\leq j\leq i$. We show that this sequence can be extended with a position $p_{i+1}\in Pos(s)$ such that  requirements (L1) and (L2) hold for $j=i+1$. Let $x:=s|_{p_i}.$ Then $\depth(\tr_{\ty+1}(x))=\depth(x)$ and, since $\depth(x)=\depth(s)-i$ and $i<\depth(s)-1$, we have $\depth(x)>1$. Therefore, we can write $x=\sigma(x_1,\ldots,x_m)$  such that $I:=\{j\in\{1,\ldots,m\}\mid \depth(x_j)\geq 1\}\neq\emptyset.$

Assume, by contradiction, that no such position $p_{i+1}$ exists. Let $q_j:=\tr_{\ty+1}(x_j)$ for all $j\in I$, and $y=\sigma(y_1,\ldots,y_m)$ where $y_j:=q_j$ if $j\in I$ and $y_j:=x_j$ otherwise. Then $y\in S^{\ty+1}\cup X(S^{\ty+1})$,
 $q_j\sim x_j$ and $\depth(q_j)<\depth(x_j)$ for all $j\in I$. It follows that $\depth(y)<\depth(x)$, and
$x\sim y$ in $\mathbb{T}(\ty+1)$, by Lemma \ref{lele}. We distinguish two cases:
\begin{enumerate}
\item $y\in S^{\ty+1}$. Then $\depth(\tr_{\ty+1}(x))\leq \depth(y)<\depth(x)$, which is a contradiction.
\item $y\in X(S^{\ty+1})$. Then $y\nsim z$ for all $z\in S^{\ty+1}$ with $\depth(z)\leq\depth(y)$, because:
\begin{itemize}
\item[] If there exists $z\in S^{\ty+1}$ with $\depth(z)\leq\depth(y)$ such that $y\sim z$, then $x\sim z$ (by Lemma \ref{lema1}) and $\depth(z)<\depth(x)$, which contradicts $\depth(\tr_{\ty+1}(x))=\depth(x)$.
\end{itemize}
As $\depth(y)<\depth(x)=\depth(s|_{p_i})\leq\depth(s)$, $y$ would be introduced in $S^{\ty+1}$ instead of $s$ as the result of a failed closedness check. This also provides a contradiction.
\end{enumerate}
Thus, there exists a sequence of positions $p_0<\ldots<p_{\depth(s)-1}$ from $Pos(s)$ such that requirements (L1) and (L2) hold for all $j\in\{0,1,\ldots, \depth(s)-1\}$. 
It follows that the statement of this lemma holds for $p=p_{\depth(s)-1}.$
\section{Proof of Lemma \ref{cor4}}
Let $r$ be the representative newly introduced  in $\cQ^{\ty+1}$ at position $j$ (that is, $r=\tty_{\ty+1}(j)$), $k:=\depth(r)$, and 
$i':=\max\{i\mid \depth(\tty_{\ty}(i))\leq k\}.$ Then $j=i'+1$ and we distinguish two situations.
\begin{enumerate}
\item If $r$ replaces a representative with depth $k'$ at position $j'$ in $\cQ^\ty$ then $k<k'$, $i'< j'$ and $j=i'+1$. Thus $j\leq j'$ and
\begin{itemize}
\item if $1\leq i<j$ then $\tty_{\ty}(i)=\tty_{\ty+1}(i)$, 
\item $\depth(\tty_\ty(j))>k=\depth(\tty_{\ty+1}(j))$, 
\item if $j<i\leq j'$ then $\depth(\tty_{\ty}(i))\geq\depth(\tty_\ty(i-1))=\depth(\tty_{\ty+1}(i))$, 
\item if $j'<i\leq n$ then $\tty_\ty(i)=\tty_{\ty+1}(i)$.
\end{itemize}
Hence $\tpl(\cQ^{\ty+1})<\tpl(\cQ^\ty).$
\item Otherwise, $r$ is newly introduced at position $j=i'+1$ in $\cQ^{\ty+1}$ and all elements of $\cQ^\ty$ are preserved in $\cQ^{\ty+1}$. If $|\cQ^\ty|=m$ then
\begin{itemize}
\item if $1\leq i< j$ then $\tty_{\ty}(i)=\tty_{\ty+1}(i)$, 
\item $\depth(\tty_\ty(j))>k=\depth(\tty_{\ty+1}(j))$, 
\item if $j<i\leq m$ then $\depth(\tty_{\ty}(i))\geq\depth(\tty_\ty(i-1))=\depth(\tty_{\ty+1}(i))$,
\item $\depth(\tty_\ty(m+1))=\ell+1>\depth(\tty_{\ty+1}(m+1))$
\end{itemize}
which, again, implies  $\tpl(\cQ^{\ty+1})<\tpl(\cQ^\ty).$ \qed
\end{enumerate}
\section{Proof of Theorem \ref{tyury}}
By Lemma \ref{cor4}, $\tpl(\cQ^{\ty})>\tpl(\cQ^{\ty+1})$ and $\depth(\tty_{\ty+1}(j))\leq j$ whenever $\tty_{\ty+1}(j)$ is the state introduced in $\cQ^{\ty+1}$ by a failed closedness check. It is also easy to see that $\tpl(\cQ^\ty)\geq\tpl(\cQ^{\ty+1})$ always holds. 
Since $\tpl(\cQ^0)=(\ell+1,\ldots,\ell+1)$ and the minimum possible value of $\tpl(\cQ^\ty)$ is $(1,\ldots,1)$,  the maximum number of failed closedness checks in any sequence 
$$\tpl(\cQ^0)\geq \tpl(\cQ^1)\geq \ldots \geq\tpl(\cQ^\ty)$$ is at most $\ty\leq \sum_{j=1}^n j=n (n+1)/2.$

\section{Proof of Corollary \ref{cor3}}
If $s$ is introduced in $S^{\ty+1}$ by a failed closedness check then $s\nsim t$ for all $t\in S^{\ty}$ with $\depth(t)\leq\depth(s).$ Therefore, $s\in\cQ^{\ty+1}\setminus \cQ^{\ty}.$ Furthermore, from the proof of the previous lemma we know there exists a sequence $$p_0<p_1<\ldots<p_{\depth(s)-1}$$ of positions from $Pos(s)$ with $\depth(\tr_{\ty+1}(s|_{p_j}))=j$ for all $0\leq j<\depth(s).$ Since $\tr_{\ty+1}(s|_{p_j})\in\cQ^{\ty+1}$ for all $0\leq j<\depth(s)$ and $\tr_{\ty+1}(s)=s$, we have
$$\underbrace{\tr_{\ty+1}(s|_{p_{\depth(s)-1}})\prec_{\mathtt{T}}\ldots\prec_{\mathtt{T}}\tr_{\ty+1}(s|_{p_1})\prec_{\mathtt{T}}\tr_{\ty+1}(s|_{p_0})}_{\depth(s)\ \text{elements}}=s$$
 we conclude that, if $s=\tty_\ty(j)$, then $\depth(s)\leq j$. 
\section{Proof of Theorem \ref{conscheck}}
First, we prove by induction on $i$, $1\leq i\leq k$, that $\{C_1,\ldots,C_i\}$ induces a partition of $\cQ^\ty$. 
In the base case, $i=1$, $\{C_1,\ldots,C_i\}=\{C_1\}=\{\hole\}$, and the statement of the lemma is obviously true.
In the induction step, we assume that $\{C_1,\ldots,C_i\}$ induces a partition $Q_1,$ \ldots, $Q_m$ of $\cQ^\ty$.  Let $M_2:=\{Q_i\mid |Q_i| >1\}$, and $M:=\bigcup_{Q_i\in M_2}Q_i$. As all pairs of elements in $M$ are $\ell$-distinguished by some element of $C_{i+1},\ldots,C_k$ and $\depth_\hole(C_{i+1})\leq\depth_\hole(C_j)$ for all $i<j\leq k$, the depth of any term contained in $M$ is at most $\ell-\depth_\hole(C_{i+1}).$ Thus $T(C_{i+1}[t])\in\{0,1\}$ for all $t\in M$, and therefore $\{C_1,\ldots,C_i,C_{i+1}\}$ induces a partition of $\cQ^\ty$.

Let $C_{i_1},\ldots,C_{i_k}$ be the order in which the contexts were added to $E$ by failed consistency checks. Because every $C_{i_{p+1}}$ $\ell$-distinguishes some elements of $\cQ^\ty$ that were not $\ell$-distinguished by any of $C_{i_j}$ with $1\leq j\leq p$, we conclude that $\cE^\ty$ induces a partition of $\cQ^\ty$ into at least $k$ classes. 

\section{Proof of Corollary  \ref{cor5}}
 Let $k$ be the number of elements of $\cE^\ty$, and $m$ be the number of classes in the partition of $\cQ^\ty$ induced by $\cE^\ty$. By Lemma \ref{conscheck}, $k\leq m$. Since $m\leq n$,  we conclude that $k\leq n$.
\section{Proof of Lemma \ref{tmy4}}
The proof is by induction on the execution step $\ty$ of the algorithm.

In the base case, assume $\cQ^0$ has $i=2$ elements. Then $E^0=\{\hole\}$ and $\depth_\hole(\hole)=0=i-2.$
In the induction case, we assume that the result holds at some step $\ty$ in the execution of the algorithm, and we prove that the result holds at the next step $\ty+1.$ 

If step $\ty$ is a failed closedness check or a failed structural equivalence query, then $E^{\ty+1}=E^{\ty}$, and $\cQ^{\ty+1}$ has at least the same number of elements as $\cQ^{\ty}.$ Therefore, the result will hold at step $\ty+1.$

Otherwise, the execution step $\ty$ is a failed consistency check. Let $s_1,s_2\in S^\ty,$ $C\in E^\ty$, and $C_1\in\sigma_\hole\pair{S^\ty}$ be the values for which this failed consistency check is performed. Then $E^{t+1}=E^{\ty}\cup\{C[C_1]\}$. We distinguish two cases:
\begin{enumerate}
\item $s_1$ and $s_2$ are  $\ell$-distinguished by some $C'\in E^{\ty}$, but  $\depth_\hole(C')>\depth_\hole(C[C_1]).$ Then  $\depth_\hole(C'')\leq \max\{\depth_\hole(C')\mid C'\in E^{\ty}\}$ for all $C''\in E^{\ty+1}$. Since $\max\{\depth_\hole(C')\mid C'\in E^{\ty}\}\leq i-2$  by induction hypothesis, and $i=|\cQ^{\ty}|\leq |\cQ^{\ty+1}|$, we learn that  $\depth_\hole(C'')\leq |\cQ^{\ty+1}|-2$ for all $C''\in E^{\ty+1}$.
\item $s_1$ and $s_2$ are not $\ell$-distinguished by any element of $E^{\ty}$. If $\depth_\hole(C[C_1])\leq\max \{\depth_\hole(C')\mid C'\in E^{\ty}\}$, the result will hold at step $\ty+1$. Otherwise, by induction hypothesis $\depth_\hole(C)\leq i-2$ and thus $\depth_\hole(C[C_1])\leq i-1.$ Let $R:=\cQ^{\ty}\cup\{s_1,s_2\}.$ Since $s_1\sim_\ell s_2$ at step $\ty$, at least one of $s_1$ and $s_2$ is not contained in $\cQ^{\ty}$, thus $R$ will have at least $|\cQ^{\ty}|+1=i+1$ elements. As $C[C_1]$ $\ell$-distinguishes  $s_1$ and $s_2$ and $\depth_\hole(C[C_1])\leq\max \{\depth_\hole(C')\mid C'\in E^{\ty}\}$, $\depth_\hole(C'[s_1])\leq \ell$ and $\depth_\hole(C'[s_2])\leq \ell$ for every $C'\in E^{\ty}.$ Thus, both $E^{\ty}$ and $E^{\ty+1}$ will induce a partition of $R$. As $s_1\nsim_\ell s_2$ at step $\ty$, but $s_1$ and $s_2$ are $\ell$-distinguished by $C[C_1]$ at step $\ty+1$, $E^{\ty+1}$ will partition $R$ into at least $|\cQ^{\ty}|+1$ classes. Thus, $|\cQ^{\ty+1}|\geq i+1$. Hence $\depth_\hole(C'')\leq i-1=(i+1)-2\leq|\cQ^{\ty+1}|-2$ for all $C''\in E^{\ty+1}.$ 
 \end{enumerate}
\section{Proof of Lemma \ref{lups}}
Suppose $\ty'$ is the first execution step when $\xpm_{\ty'}(j)\neq n-1.$ This means that $\ty'$ is the first execution step from where on  we distinguish at least $j+1$ representatives in the observation table. Therefore, at the previous step $\ty'-1$, $|\cQ^{\ty'-1}|\leq j$ and so, by Lemma \ref{tmy4}, $\depth_\hole(C)\leq j-2$ for all $C\in E^{\ty'-1}.$ Thus, $\depth_\hole(C')\leq j-1$ for all $C'\in E^{\ty'}$, and in particular $\xpm_{\ty'}(j)\leq j-1.$ Since it it obvious that $\xpm_\ty(j)\leq\xpm_{\ty'}(j)$ whenever $\ty \geq\ty'$, we conclude that $\xpm_\ty(j)\leq j-1$ whenever $2\leq j\leq n$ and $\xpm_\ty(j)\neq n-1.$
\section{Proof of Theorem \ref{inconsistentthm}}
It is easy to see that  $\tpl(\cQ^\ty)\geq \tpl(\cQ^{\ty+1})$ holds for every execution step $\ty$. Moreover, if the $\ty$-th execution step is a failed consistency check then a  context $C$ is newly added to $\cE^\ty$ in order to produce $\cE^{\ty+1}$. The context $C$ will $\ell$-distinguish two elements $s_1,s_2\in S^\ty$ that were not $\ell$-distinguished before or had been $\ell$-distinguished by some $C'\in E^\ty$ with $\depth_\hole(C')>\depth_\hole(C).$ Since $\depth(\tr_{\ty+1}(s_1))\leq\depth(s_1)$ and $\depth(\tr_{\ty+1}(s_2))\leq\depth(s_2)$, $C$ will $\ell$-distinguish two elements of $S^\ty$ that were not $\ell$-distinguished before or were $\ell$-distinguished by a context with bigger hole-depth. Therefore $\tpl(\cQ^{\ty+1})<\tpl(\cQ^\ty)$ if $\ty$ is a failed  consistency check.

Note that $\tpl(\cQ^0)=(0,n-1,\ldots,n-1)$ and the minimum possible value of $\tpl(\cA^\ty)$ is $(0,1,\ldots,1)$.  Also, by Lemma \ref{lups}, $\xpm_\ty(j)\leq j-1$ whenever $\xpm_\ty(j)\neq n-1$ for $2\leq j\leq n$. Therefore, any run of the algorithm performs at most $\sum_{j=2}^n(j-1)={n(n-1)}/{2}$
failed consistency checks. 
\section{Proof of Theorem \ref{failstruct}}
Algorithm $LA^\ell$ performs a failed structural equivalence query when the observation table $\mathbb{T}(\ty)$ is closed, consistent, and has less than $n$ states (by Corollary \ref{cor2} of Theorem \ref{thm1}). Suppose the algorithm performed a 
 failed structural equivalence query for $\cA(\mathbb{T}(\ty))$ which rendered the counterexample $t$. After extending the $S$-component of observation table $\mathbb{T}(\ty)$ with all subterms of $t$ that were not yet there, the algorithm constructs a new observation table $\mathbb{T}(\ty')$ which is closed and consistent. Since $t\in S^{\ty'}$, $\hole\in E^{\ty'}$, and $T(\hole[t])$ in the table $\mathbb{T}(\ty)$ differs from $T(\hole[t])$ in the table $\mathbb{T}(t')$, the automata $\cA(\mathbb{T}(t))$ and $\cA(\mathbb{T}(t'))$ are not equivalent with respect to $\ell$ (that is, $\cL(\cA(\mathbb{T}(\ty)))_{[\ell]}\neq\cL(\cA(\mathbb{T}(\ty')))_{[\ell]}$). Therefore, by Theorem \ref{oops}, the automaton $\cA(\mathbb{T}(\ty'))$ must have at least one more state than 
$\cA(\mathbb{T}(\ty'))$. Since the number of states is increased by every failed structural equivalence query and can not exceed $n-1$, the number of failed structural equivalence queries performed by algorithm $LA^\ell$ is at most $n$. 
\end{document}